\documentclass[DIV12]{scrartcl}

\usepackage{microtype}
\usepackage{amsmath}
\usepackage{amssymb}
\usepackage{booktabs}
\usepackage{tabularx}
\usepackage{xcolor}
\usepackage{tikz}
\usepackage{stmaryrd}
\usepackage{xspace}
\usepackage{paralist}
\usepackage{balance}
\usepackage{enumitem}
\usepackage[bibliography=common]{apxproof}

\usepackage[utf8]{inputenc}
\usepackage[hidelinks,pdfborder={0 0 0},breaklinks]{hyperref}
\usepackage{breakurl}

\hypersetup{
    colorlinks,
    linkcolor={red!50!black},
    citecolor={blue!50!black},
    urlcolor={blue!30!black}
}

\title{Tractable Lineages on Treelike Instances:\\Limits and Extensions}

\author{
\begin{tabular}[t]{c}
Antoine Amarilli \\
{\normalfont LTCI, CNRS, Télécom ParisTech, Université Paris-Saclay} \\
{\normalfont antoine.amarilli@telecom-paristech.fr} \\[0.5em]
Pierre Bourhis \\
{\normalfont CRIStAL, UMR 9189, CNRS, Université Lille 1} \\
{\normalfont pierre.bourhis@univ-lille1.fr} \\[0.5em]
Pierre Senellart \\
{\normalfont
LTCI, CNRS, Télécom ParisTech, Université Paris-Saclay} \\
{\normalfont \& IPAL, CNRS, NUS} \\
{\normalfont pierre.senellart@telecom-paristech.fr} \\[0.5em]
\end{tabular}
}

\date{}

\newcommand{\tw}{\mathsf{tw}}
\newcommand{\G}{\mathcal{G}}
\newcommand{\I}{\mathcal{I}}

\makeatletter
\newcommand*{\defeq}{\mathrel{\rlap{%
  \raisebox{0.3ex}{$\m@th\cdot$}}%
  \raisebox{-0.3ex}{$\m@th\cdot$}}%
  =}
\makeatother
\newcommand{\fpsp}{\smash{\text{FP}^{\mathrm{\#P}}}}
\newcommand{\fpspo}{\text{FP}^{\mathrm{\#P}_1}}
\renewcommand{\sp}{\Sigma^{\mathrm{P}}}

\DeclareMathOperator{\dom}{dom}

\newcommand{\true}{1}
\newcommand{\false}{0}

\newcommand{\calG}{\mathcal{G}}

\newcommand{\calI}{\mathcal{I}}

\newcommand{\calP}{\mathcal{P}}

\newcommand{\calX}{\mathcal{X}}

\newcommand{\LCf}{\text{L}}
\newcommand{\RCf}{\text{R}}

\newcommand{\LC}{L}
\newcommand{\RC}{R}

\newcommand{\card}[1]{\left|#1\right|}

\let\oldPr\Pr
\renewcommand\Pr{\oldPr\nolimits}

\newcommand{\posbool}[1]{\mathrm{PosBool}[#1]}

\newcommand{\cqneq}{\smash{$\mathrm{CQ}^{\neq}$}\xspace}
\newcommand{\ucqneq}{\smash{$\mathrm{UCQ}^{\neq}$}\xspace}
\newcommand{\ucq}{$\mathrm{UCQ}$\xspace}
\newcommand{\rpqneq}{\smash{$\mathrm{C2RPQ}^{\neq}$}\xspace}

\newcommand{\arity}[1]{\mathrm{arity}(#1)}

\newcommand{\dd}{\mathrm{d}}

\newcommand{\h}{\mathrm{h}}

\renewcommand{\i}{\mathrm{i}}

\newcommand{\p}{\mathrm{p}}

\newcommand{\NN}{\mathbb{N}}

\newcommand{\deft}[1]{\emph{#1}}

\newcommand{\ofv}{\mathrm{ofv}}

\newcommand{\NNs}{\mathbb{N}_{>0}}

\newcommand{\defo}[1]{\emph{#1}}

\renewcommand{\phi}{\varphi}
\renewcommand{\leq}{\leqslant}
\renewcommand{\geq}{\geqslant}

\newtheoremrep{theorem}{Theorem}[section]
\newtheoremrep{example}[theorem]{Example}
\newtheoremrep{proposition}[theorem]{Proposition}
\newtheoremrep{lemma}[theorem]{Lemma}
\newtheoremrep{corollary}[theorem]{Corollary}
\newtheoremrep{definition}[theorem]{Definition}

\begin{document}

\maketitle

\begin{abstract}
  Query evaluation on probabilistic databases is generally intractable (\#P-hard).
Existing dichotomy
results~\cite{DBLP:conf/pods/DalviS07a,dalvi2012dichotomy,DBLP:conf/pods/FinkO14}
have identified which queries are tractable (or \emph{safe}), and connected them
to tractable lineages~\cite{jha2013knowledge}. In our previous
work~\cite{amarilli2015provenance}, using different tools, we showed that query
evaluation is linear-time on probabilistic databases for arbitrary monadic
second-order queries, if we bound the \emph{treewidth} of the instance.

In this paper, we study limitations and extensions of this result. First, for
probabilistic query evaluation, we show that MSO tractability cannot extend
beyond bounded treewidth: there are even FO queries that are hard on \emph{any}
efficiently constructible unbounded-treewidth class of graphs. This dichotomy
relies on recent polynomial bounds on the extraction of planar graphs as
minors~\cite{chekuri2014polynomial}, and implies lower bounds in
non-probabilistic settings, for query evaluation and match counting in
subinstance-closed families. Second, we show how to explain our tractability
result in terms of lineage: the lineage of MSO queries on bounded-treewidth
instances can be represented as bounded-treewidth circuits, polynomial-size
OBDDs, and linear-size d-DNNFs. By contrast, we can strengthen the previous
dichotomy to lineages, and show that there are even UCQs with disequalities that
have superpolynomial OBDDs on \emph{all} unbounded-treewidth graph classes; we
give a characterization of such queries. Last, we show how bounded-treewidth
tractability explains the tractability of the \emph{inversion-free} safe
queries: we can rewrite their input instances to have bounded-treewidth.

\end{abstract}

\section{Introduction}
\label{sec:introduction}
Many applications must deal with data which may be erroneous. This makes it
necessary to extend relational database instances, to allow for
uncertain facts. One of the simplest such formalisms~\cite{suciu2011probabilistic} is
that of \emph{tuple-independent databases} (TID): each tuple in the database is
annotated with an independent probability of being present. The semantics of a
TID instance is to see it as a concise representation of a probability distribution on
standard non-probabilistic instances.

An important challenge when dealing with probabilistic data is that
data management tasks become intractable. The main one is
\emph{query evaluation}: given an input database query~$q$, for instance a
conjunctive query, and given a relational instance $I$, determine the answers to
$q$ on $I$. When $q$ is Boolean, we just ask whether $I$ satisfies $q$. The
corresponding problem in the TID setting asks for the \emph{probability}
that $I \models q$, that is, the total probability weight of the possible
subsets of the TID instance~$I$ that satisfy $q$. The query~$q$ is usually assumed to be
fixed, and we look at the complexity of this problem as a function of the
input instance (or TID) $I$, that is, the \emph{data complexity}. Sadly, while
this task is highly tractable and parallelizable (in the class~AC$^0$) in the non-probabilistic context, exact computation is generally intractable
(\#P-hard) on TIDs, even for the simple conjunctive query
$\exists x y
~ R(x)\land S(x, y)\land T(y)$. See~\cite{dalvi2007efficient}.

Faced with this intractability, two natural directions are possible. The first
is to restrict the language of \emph{queries} to focus on queries that are
tractable on \emph{all} instances, called \emph{safe}. This has proven a very fruitful
direction~\cite{DBLP:conf/pods/DalviS07a},
culminating in the dichotomy result of Dalvi and
Suciu~\cite{dalvi2012dichotomy}:
the data complexity of a given union of conjunctive queries
(UCQ) is either in PTIME or \#P-hard. More recently, 
the safe
non-repeated CQs \emph{with negation} were characterized
in~\cite{DBLP:conf/pods/FinkO14}.

The second approach is to restrict the \emph{instances}, to focus on instance
families that are tractable for \emph{all} queries in highly expressive
languages. In a recent work~\cite{amarilli2015provenance}, going through the
setting of semiring provenance~\cite{green2007provenance}, 
and using a celebrated result by Courcelle~\cite{courcelle1990monadic},
we have started to explore this direction. We showed that, for queries in MSO (\emph{monadic second-order}, a highly
expressive language capturing UCQs), data complexity is linear-time on
\emph{treelike} instances, i.e., instances of
\emph{treewidth} bounded by a constant. Of course, this result says nothing of non-treelike instances, but
covers use cases previously studied in their own right, such as
probabilistic XML~\cite{cohen2009running} (without data values).

This new direction raises several important questions:
\begin{itemize}
\item First, is this the best that one can hope for? For probability
  evaluation, could the tractability on
bounded-treewidth instances be generalized, e.g., to bounded \emph{clique-width}
instances~\cite{courcelle1993handle}, as for MSO in the non-probabilistic case?
More ambitiously, could we 
separate tractable and intractable instances with a dichotomy theorem?
\item Second, can our bounded-treewidth tractability result be explained in terms of
  \emph{lineage}? The \emph{lineage} of a query intuitively represents how it
  can be satisfied on the instance, and can be used to compute its probability:
for many fragments of safe queries~\cite{jha2013knowledge}, tractability can be
shown via a tractable representation of their lineage. In
\cite{amarilli2015provenance}, we build a bounded-treewidth circuit representation of Boolean
provenance. How does this compare to the usual lineage classes of OBDDs and
\mbox{d-DNNFs}
in knowledge compilation?
\item Third, can we link the query-based tractability approach to our
  instance-based one? Can we explain the tractability of
some safe queries by reducing them to query evaluation on treelike
instances?
\end{itemize}
This paper answers all of these questions. 

\paragraph*{Contributions}
Our \emph{first main result} (in Section~\ref{sec:probability}) is that
bounded treewidth characterizes the tractable families of graphs for MSO
queries in the probabilistic context.
More precisely, we construct a query~$q_{\h}$ for which probability evaluation
is intractable  on \emph{any} unbounded-treewidth family of graphs
satisfying mild constructibility requirements; 
query evaluation is precisely $\fpsp$-complete under randomized polynomial-time
(RP) reductions.
Thus, tractability on bounded-treewidth instances is really the best we can
get, on arity-2 signatures. Surprisingly, we
show that $q_{\h}$ can be taken to be a (non-monotone) FO query; this
is in stark contrast with non-probabilistic query evaluation~\cite{kreutzer2010lower,ganian2014lower}
where FO queries are fixed-parameter tractable under much milder
conditions than bounded treewidth~\cite{kreutzer2008algorithmic}.
This
provides the lower bound of a dichotomy, the upper bound being our result
in~\cite{amarilli2015provenance}.

In Section~\ref{sec:consequences}, we explain how this dichotomy result
can be adapted to non-probabilistic MSO query evaluation and match counting on
subgraph-closed graph families. While the necessity of bounded-treewidth for
non-probabilistic query evaluation was studied
before~\cite{kreutzer2010lower,ganian2014lower}, our use of a recent
polynomial bound on grid minors~\cite{chekuri2014polynomial} allows us to obtain
stronger results in this context, which we review. Our work thus answers the
conjecture of~\cite{grohe2008logic} (Conjecture~8.3) for MSO,
which~\cite{kreutzer2010lower} answered for MSO$_2$, under similar
complexity-theoretic assumptions.

In Section~\ref{sec:lineages}, we move from probability evaluation to
the computation of tractable \emph{lineages}.
Our tractability result in~\cite{amarilli2015provenance} computes
a bounded-treewidth lineage of linear size for MSO queries on bounded-treewidth
instances. We revisit this
upper bound and show that we can compute an OBDD lineage
of polynomial size (by results in~\cite{jha2012tractability}) and a
\mbox{d-DNNF}
lineage of linear size (a new result). We show that on
bounded-\emph{pathwidth} instances (a notion more restrictive than that of
bounded-treewidth), we obtain a bounded-pathwidth lineage, and
hence a constant-width OBDD (by~\cite{jha2012tractability}). Further,
all these
representations can be efficiently constructed.

We then reexamine the choice of representing provenance as a \emph{circuit}
rather than a formula, because this is unusual in the semiring provenance
context of~\cite{amarilli2015provenance}.
We show in Section~\ref{sec:formulae} that some of the previous tractability results for
lineage representations \emph{cannot} extend to formula representations,
via conciseness bounds on
Boolean circuits and formulae. This sheds some light on the 
conciseness gap between circuit and formula representations of lineage.

We then move in Section~\ref{sec:meta} to our \emph{second main result},
which applies to tractable \emph{OBDD lineages} rather than tractable query evaluation.
It shows a dichotomy on arity-2 signatures,
for the weaker query language of \emph{UCQs with disequalities}:
while bounded-treewidth instances admit efficient OBDDs for such queries, any constructible
unbounded-treewidth instance family must have superpolynomial OBDDs for some
query (depending only on the signature). 

Last, in Section~\ref{sec:safe}, we connect our approach to query-based
tractability conditions~\cite{DBLP:conf/pods/DalviS07a,dalvi2012dichotomy}. We show that,
for safe UCQs that admit a concise OBDD representation (that is,
precisely inversion-free UCQs from~\cite{jha2013knowledge}), one can
rewrite any instance to a bounded-treewidth instance (actually, to a
bounded-pathwidth one), such that the query lineage, and hence the query
probability, remain the
same. Thus, in this sense, safe queries are tractable because their
input instances may as well be bounded-pathwidth.

\paragraph*{Related work}
Bounded-treewidth has been shown to be a sufficient condition for
tractability of query evaluation (this is by Courcelle's
theorem~\cite{courcelle1990monadic}, generalized to arbitrary relational
structures in~\cite{flum2002query}), counting of query
matches~\cite{arnborg1991easy}, and probabilistic query
evaluation~\cite{amarilli2015provenance}.

For MSO query evaluation on non-probabilistic instances, bounded-treewidth is known not to be
necessary, e.g., query evaluation is tractable assuming
bounded \emph{clique-width}~\cite{courcelle1993handle}. FO query evaluation is
tractable assuming milder conditions~\cite{kreutzer2008algorithmic}.
Two separate lines of work investigated the necessity of
bounding the treewidth of instances to ensure the tractability of other data management
tasks.

First, in~\cite{DBLP:conf/focs/Marx07,marx2010can}, Marx
shows that treewidth-based algorithms for binary constraint-satisfaction
problems (CSP) are, assuming the exponential-time
hypothesis, \emph{almost optimal}: they can only be improved by a
logarithmic factor. These works do not rely on the graph minor
theorem~\cite{robertson1986graph5} as
we do, as they preceded the results of~\cite{chekuri2014polynomial} that
provide polynomial bounds on the size of grid minors: see the
discussion in the Introduction of~\cite{marx2010can}.
Instead, 
they characterize 
high treewidth via embeddings
of low \emph{depth}. The results
of~\cite{DBLP:conf/focs/Marx07,marx2010can} were further
applied to inference in
undirected~\cite{DBLP:conf/uai/ChandrasekaranSH08} and
directed~\cite{DBLP:conf/ecai/KwisthoutBG10} \emph{graphical models}. All these
works are specific to the setting and problem that they study, namely CSP and inference.

Second, another line of
work~\cite{makowsky2003tree,kreutzer2010lower,ganian2014lower} has shown
necessity of bounded treewidth when a class of graphs is closed under
some operations: extracting topological minors
in~\cite{makowsky2003tree}, extracting subgraphs
in~\cite{kreutzer2010lower}, and extracting subgraphs and vertex
relabeling in~\cite{ganian2014lower}. This requires that there are sufficiently
many instances of high treewidth, through notions of \emph{strong
unboundedness}~\cite{kreutzer2010lower} and \emph{dense
unboundedness}~\cite{ganian2014lower}. We strengthen the results
of~\cite{ganian2014lower} in Section~\ref{sec:alternate} of this paper, using our
techniques. None of these works consider probabilistic evaluation or
match counting, which we do here.

Other related work is discussed throughout the paper, where relevant; in
particular works related to lineages in Sections~\ref{sec:lineages}
to~\ref{sec:meta} and
to safe queries in Section~\ref{sec:safe}.

\medskip

The next section (Section~\ref{sec:prelim}) presents preliminaries, and 
Section~\ref{sec:problem} gives our formal problem statement.
We then move to
our new results in Section~\ref{sec:probability} onwards.

This paper does not include the full proofs of all results. The appendix
presents the full proofs of the results of Sections~\ref{sec:lineages},
\ref{sec:formulae}, and~\ref{sec:safe}, as well as the proof of
Theorem~\ref{thm:ganianvariant}. The full proofs of the other results can be
found in Chapter~6 of~\cite{thesis}. Note that the proof of
Theorem~\ref{thm:ganianvariant} given in the appendix of the present paper is a
corrected version of the erroneous proof of the same result given
in~\cite{thesis}.

\section{Preliminaries}
\label{sec:prelim}
\vspace{-.6em}
\paragraph*{Instances}
A \emph{relational signature} $\sigma$ is a set of relations $R, S,
T, \ldots$, each having an \emph{arity} denoted $\arity{R}\in\NNs$.
The signature $\sigma$ is 
\emph{arity-$k$} if $k$ is the
maximum arity of a relation in~$\sigma$.

A \emph{relational instance} (or simply $\sigma$-instance or instance) $I$ is a finite set of
ground \emph{facts} on the signature $\sigma$,
and a \emph{class} or
\emph{family} of instances
$\I$ is just a (possibly infinite) set of instances. A \emph{subinstance} of~$I$ is a subset
of its facts. We follow the \emph{active domain semantics}, where the \emph{domain}
$\dom(I)$ of~$I$ is the finite set of elements that occur in facts.
Hence, for $I' \subseteq I$, $\dom(I')$ is the (possibly strict) subset of
$\dom(I)$ formed of the elements that occur in facts of~$I'$. The
\emph{size} of~$I$, denoted $|I|$, is its number of facts.

A \emph{homomorphism}
from a $\sigma$-instance~$I$
to 
a $\sigma$-instance~$I'$ 
is a function $h: \dom(I) \to \dom(I')$
such that, for all
$R(a_1,\dots,a_k)\in I$, we have $R(h(a_1),\dots,h(a_k))\in I'$.
A homomorphism is an \emph{isomorphism} if it is bijective and its inverse is
also a homomorphism.

\paragraph*{Graphs} Throughout the paper, a \emph{graph} will always be
undirected, simple, and unlabeled, unless otherwise specified.
Formally, we see a graph $G$ as an instance of the \emph{graph signature} with a
single predicate $E$ of arity~$2$ such that: \begin{inparaenum}[(i)]
\item $\forall x\, E(x,x)\notin G$; and
\item $\forall xy\, E(x,y)\in G\Rightarrow E(y,x)\in G$. 
\end{inparaenum}
As we follow the active domain semantics, this implies that we disallow \emph{isolated vertices} in graphs.
The facts of $G$ are called \emph{edges}.
The set of \emph{vertices} (or \emph{nodes}) of a graph~$G$, denoted $V(G)$,
is its domain.
Two vertices $x$ and~$y$ of a graph~$G$ are \emph{adjacent} if $E(x,y)\in
G$, $x$ and $y$ are then called the \emph{endpoints} of the edge, and the edge
is \emph{incident} to them; two edges are \emph{incident} if they share a vertex.

The \emph{degree} of a vertex $x$ is the number of its adjacent
vertices.
For $k\in\NN$, a graph is $k$-\emph{regular} if all vertices have
degree~$k$. More generally, it is $K$-\emph{regular}, where $K$ is a
finite set of integers, if every vertex has
degree~$k$ for some $k\in K$. Finally, a graph is \emph{degree-$k$} if $k$ is the
maximum of the degree of all its vertices, i.e., if it is
$\{1, \ldots,
k\}$-regular.
A graph is \emph{planar} if it can be drawn on the plane
without edge crossings, in the standard sense~\cite{diestel}.

A \emph{path} of length~$n\in\NNs$ in a graph~$G$ is a set of edges 
$\{E(x_0,x_1),E(x_1,x_2),\dots,E(x_{n-1},x_n)\}$ that are all in~$G$; the
path is \emph{simple} if all 
$x_i$'s are distinct. 
A \emph{cycle} is a path of length $n \geq 3$ where all vertices are distinct
except that $x_0 = x_n$;
a graph is \emph{cyclic} if it has a cycle. A graph is \emph{connected} if there
is a path from every vertex to every other vertex. 
A \emph{subdivision} of a graph~$G$ is a graph obtained by
replacing each edge by an arbitrary non-empty simple path (every node
on this path being fresh except the endpoints of the original edge). 
  
\paragraph*{Treewidth and pathwidth}
A \emph{tree}~$T$ is an acyclic connected graph (remember that graphs are
undirected). A \emph{tree decomposition} of a graph~$G$ is a tree with a
labeling function $\lambda$ from its
nodes (called \emph{bags}) to
sets of vertices of~$G$, ensuring: \begin{inparaenum}[(i)]
\item for every edge $E(u,v)\in G$, there is a bag $n\in V(T)$ such that
  $\lambda(n)$ contains both $u$ and~$v$;
\item for every node $u$ of~$G$, the subtree of~$T$ formed of all bags whose
  $\lambda$-image contains~$u$ must be connected.
\end{inparaenum}
The \emph{width} of $T$ is $\max_{n\in
V(T)}|\lambda(n)|-1$. The \emph{treewidth} of a graph~$G$, denoted~$\tw(G)$,
is the minimum width of any
tree decomposition of~$G$. 

The \emph{treewidth} of a relational instance $I$, denoted $\tw(I)$, is defined as usual as the
treewidth of its \emph{Gaifman graph}, namely, the graph on the \emph{domain}
$\dom(I)$ of $I$ that connects any two elements that co-occur in a fact.
When the signature is arity-2, we can see an instance $I$ as a labeled graph,
and the treewidth of~$I$ is then exactly the treewidth of this graph.

A \emph{path decomposition} is a tree decomposition where $T$ is also a path.
The \emph{pathwidth} of a graph~$G$ is the minimum width of any path
decomposition of the graph. The \emph{pathwidth} of a relational instance
is the pathwidth of its Gaifman graph.

\paragraph*{Queries}
A \emph{query} on the signature~$\sigma$ is a formula in second-order
logic over predicates of~$\sigma$ and equality, with its standard semantics.
All queries that we consider have no constants; unless otherwise
specified, they are \emph{Boolean}, i.e.,
they have no free variable.
We write $I \models q$ whenever an instance $I$ satisfies the query $q$. We will
be especially interested in the language FO of first-order logical
sentences (where second-order quantifications are disallowed) and the
language MSO of monadic second-order logical sentences (where the only
second-order quantifications are over unary predicates).

We will also consider the language CQ of \emph{conjunctive queries}, i.e.,
existentially quantified conjunctions of atoms over the signature;
the language \cqneq{}
of conjunctive queries where additional atoms of the form $x\neq y$
(called \emph{disequality} atoms) are
allowed, where
$x$ and $y$ are variables appearing in some regular atom; the language UCQ
of
\emph{union of conjunctive queries}, namely, disjunctions of CQs; the language \ucqneq{} of
disjunctions of \cqneq{} queries. The size $\card{q}$ of a \ucqneq{} query~$q$
is its total number of atoms, i.e., the sum of the number of atoms in
each \cqneq.

A \emph{homomorphism} from a CQ $q$ to an instance $I$ is a mapping $h$ from the
variables of $q$ to $\dom(I)$ such that for each atom $R(x_1, \ldots, x_k)$ of~$q$ we
have $R(h(x_1), \ldots, h(x_k)) \in I$. For \cqneq{}
queries, we require that $h(x) \neq h(y)$ whenever $q$ contains the disequality
atom $x \neq y$. A \emph{homomorphism} from a \ucqneq $q$ to~$I$ is a
homomorphism from some disjunct of~$q$ to~$I$: 
it witnesses that $I \models q$.
A \emph{match} of a \ucqneq $q$ on an instance~$I$ is a subset of~$I$ which is the image of a
homomorphism from $q$ to $I$; a \emph{minimal match} is a match that is minimal
for inclusion.

A query is \emph{monotone} if 
$I\models q$ and $I \subseteq I'$ imply $I'\models q$
for any two instances $I, I'$. A query is \emph{closed under
homomorphisms} if we have $I' \models q$ whenever $I\models
q$
and $I$ has a homomorphism to~$I'$, for any $I$ and $I'$. UCQ is an example of query class that is both
monotone and closed under homomorphisms, while \ucqneq{} is monotone but
not closed under homomorphisms.

\section{Problem Statement}
\label{sec:problem}
We study the problem of \emph{probability evaluation}:

\begin{definition}
  Given an instance $I$, a \emph{probability valuation} is a function $\pi$
  that maps each fact of $I$ to a value\footnote{All non-integer numbers are rational
  numbers represented as the pair of their numerator and denominator.} in $[0, 1]$. A 
  probability valuation
  defines a \emph{probability distribution} on subinstances of $I$, which we
  also write $\pi$ by a slight abuse of notation. The distribution $\pi$ is
  intuitively obtained by seeing each fact $F$ as kept with probability $\pi(F)$
  and removed with probability $1 - \pi(F)$, all such choices being independent.
  Formally, the probability of $I' \subseteq I$ in this distribution is:
  \[
    \pi(I') \defeq \prod_{F \in I'} \pi(F) \prod_{F \in I \backslash I'} (1 -
    \pi(F))
  \]

  The \emph{probability evaluation problem} for a query $q$ on a class $\I$ of
  relational instances asks, given an instance $I \in \I$ and a
  probability valuation $\pi$ on $I$, what is the probability that $q$ holds
  according to the probability distribution, i.e., it is the problem of
computing $\pi(q, I) \defeq \sum_{I' \subseteq I \text{ such that } I' \models q} \pi(I')$.
\end{definition}

In other words, probability evaluation asks for the probability of~$q$ over a TID
instance defined by~$I$ and~$\pi$.
Note that we only consider classes $\I$ of instances with no associated
probabilities, and the probability valuation $\pi$ is given as an additional
input --- it is not indicated in~$\I$. The complexity of the probability
evaluation problem will always be studied in \emph{data complexity}: the query
$q$ and class $\I$ is fixed, and the input is the instance $I \in \I$ and the
probability valuation.

We also explore the problem of computing \emph{tractable lineages} (or
\emph{provenance}), defined and studied
in Section~\ref{sec:lineages} onwards.

\medskip

We rely on results of~\cite{amarilli2015provenance} that show the tractability in data complexity of
provenance computation and probability evaluation on treelike (i.e.,
bounded-treewidth) instances.
This holds for \emph{guarded
second-order} queries, but as such queries collapse to MSO under bounded
treewidth~\cite{gradel2002back}, we always use MSO queries here.
First, \cite{amarilli2015provenance} shows that we can construct Boolean circuits that represent the provenance
of MSO queries on treelike instances; we can also construct monotone circuits for
monotone queries.
The results also apply to other semirings, but this will not be our focus here.
Second, \cite{amarilli2015provenance} shows that probability evaluation
is then tractable, namely, \emph{ra-linear}:
in linear time up to the (polynomial) cost of arithmetic operations.

Our goal is thus to investigate to what extent we can generalize the following
tractability result from~\cite{amarilli2015provenance}:

\begin{theorem}[\cite{amarilli2015provenance}]
  \label{thm:main}
  For any signature~$\sigma$, for any (monotone) MSO query $q$, for
  any $k \in \NN$, there is an algorithm which, given an input instance $I$ of
  treewidth $\leq k$:
  \begin{itemize}
  \setlength\itemsep{0pt}
    \item Computes a (monotone) Boolean provenance circuit of $q$ on $I$, in
      linear time in $I$;
    \item Given a probability valuation of $I$, computes the probability of $q$
      on $I$, in ra-linear time.
  \end{itemize}
\end{theorem}

We first focus on the second point (probability evaluation) in Section~\ref{sec:probability}, followed by a
digression about non-probabilistic evaluation in Section~\ref{sec:consequences}.
We then study the
first point (lineages) in Sections~\ref{sec:lineages}--\ref{sec:meta}. We close with a
connection to safe queries in Section~\ref{sec:safe}.

\section{Probability Evaluation}\label{sec:probability}
This section studies whether we can extend
the above tractability result by lifting the
bounded-treewidth requirement.
We answer in the negative by a \emph{dichotomy result} on arity-two signatures:
there are queries for which probabilistic evaluation is tractable on
bounded-treewidth families but is intractable on \emph{any} efficiently
constructible unbounded-treewidth
family. A first technical issue is to formalize what we mean by
\emph{efficiently} constructible. We use the following notion:

\begin{definition}
  We call $\I$
  \emph{treewidth-constructible} if   for all $k \in \NN$,
  if $\I$ contains instances of treewidth $\geq k$, we can construct
  one
  in polynomial time given~$k$ written in unary\footnote{The requirement that $k$
    be given in unary rather than in binary means that \emph{more} instance
    families are
    treewidth-constructible, so treewidth-constructibility in this sense is a
  weaker assumption than if the input $k$ could be written in binary.}.
\end{definition}

In particular, this implies that $\I$ must contain a subfamily of
unbounded-treewidth instances that are small, i.e., have size polynomial in
their treewidth.
We discuss the impact of this choice of definition, and alternate
definitions of \emph{efficiently} constructible instances, in
Section~\ref{sec:consequences}.

A second technical issue is that we need to restrict to signatures of arity~$2$. We
will then show our
dichotomy for \emph{any} such signature. This suffices to
show that our Theorem~\ref{thm:main} cannot be generalized: its
generalization should apply to any signature, in particular arity-2 ones. Yet, we do
not know whether the dichotomy applies to signatures of arity $>2$.

Our \emph{main result} on probability evaluation is as follows. In this result,
$\fpsp$ is the class of function problems which can be solved in PTIME
with a deterministic Turing machine having access to a \#P-oracle, i.e., an oracle for
counting problems that can be expressed as the number of accepting paths for a
nondeterministic PTIME Turing machine.
\begin{theorem}\label{thm:dichotomy}
  Let $\sigma$ be an arbitrary arity-2 signature.
  Let $\I$ be a treewidth-constructible class of $\sigma$-instances.
  Then the following dichotomy holds:
  \begin{itemize}
  \setlength\itemsep{0pt}
    \item If there is $k\in\mathbb{N}$ such that $\tw(I)\leq k$ for every $I\in\I$,
      then for every MSO
      query $q$, the probability evaluation problem for
      $q$ on instances of~$\I$ is solvable in ra-linear time.
    \item Otherwise, there is an FO query $q_{\h}$ (depending on~$\sigma$ but
      not on $\I$) such that
      the probability evaluation problem
      for~$q_{\h}$ on $\I$ is 
      $\fpsp$-complete under randomized polynomial time (RP) reductions.
  \end{itemize}
\end{theorem}

The first part of this result is precisely the second point of Theorem~\ref{thm:main}.
We thus sketch the proof of the hardness result of the second part.
Pay close attention to the statement: while 
some FO queries (in particular, unsafe CQs~\cite{dalvi2012dichotomy})
may have $\fpsp$-hard
probability evaluation when \emph{all} input instances are allowed,
our goal here is to build a query that is hard even when input instances are
restricted to \emph{arbitrary families} satisfying our conditions, a much harder claim.

We reduce from the problem of counting graph \emph{matchings}, namely, the
number of edge subsets of a graph that have no pair of incident edges. This problem is
known to be \#P-hard on 3-regular planar graphs~\cite{xia2007computational}. We
define a FO query $q_\h$ that tests for matchings on such graphs (encoded in a
certain way), and we rely on the connection between probability evaluation
and model counting so that the probability of~$q_\h$ on (an encoding of) a
graph~$G$ reflects its number of matchings.

The main idea is that 3-regular planar graphs can be \emph{extracted} 
from our family $\calI$, using the
following notion:
\begin{definition}
  \label{def:minor}
  An 
  \emph{embedding} of a graph~$H$ in a graph~$G$ is an
  injective mapping $f$ from the vertices of $H$ to the vertices of $G$ and a
  mapping~$g$ that maps the edges $(u, v)$ of $H$ to paths in $G$ from $f(u)$ to
  $f(v)$, with all paths being vertex-disjoint.
  A graph $H$ is a \emph{topological minor} of a graph $G$ if there is an
  embedding of~$H$ in~$G$.
\end{definition}

We then use the following lemma, that rephrases the recent polynomial
bound~\cite{chekuri2014polynomial} on Robertson and Seymour's grid minor
theorem~\cite{robertson1986graph5}
to the realm of topological minors; in so doing, we use the folklore observation that a
degree-3 minor of a graph is always a topological minor:

\begin{lemma}\label{lem:extract}
  There is $c \in \mathbb{N}$ such that
  for any degree-3 planar graph $H$,
  for any graph $G$ of
  treewidth $\geq \card{V(H)}^c$,
  $H$ is a 
  topological
  minor of $G$ and an embedding of~$H$
  in $G$ can be computed in randomized polynomial time in~$\card{G}$.
\end{lemma}

Hence, intuitively, given an input 3-regular planar graph~$G$ (the input to the hard problem),
we can extract it in
randomized polynomial-time (RP) as a topological minor of (the Gaifman graph of)
an instance~$I$ of our family $\calI$ that we obtain using treewidth-constructibility. Once it is extracted, we
show that, by choosing the right probability valuation for $I$,
the probability of $q_\h$ on $I$
allows us
to reconstruct the answer to the original
hard problem, namely, the number of matchings of~$G$.
The minor extraction step is what complicates the design of $q_\h$, as
$q_\h$ must then test for matchings in a way which is \emph{invariant under
subdivisions}: this is especially tricky in FO as we can only make local tests.

\paragraph*{Choice of hard query}
Not only is our query $q_{\h}$ independent from the class of
instances~$\I$, but it is also an FO query, so, in the \emph{non-probabilistic}
setting, its data complexity on any instance is in AC$^0$. In fact, our choice of~$q_{\h}$ has
also \emph{linear-time} data complexity: one can determine in linear time in an input
instance~$I$ whether $I \models
q_{\h}$.
This contrasts sharply with the
$\fpsp$-completeness (under RP reductions)
of \emph{probability evaluation} for $q_\h$ 
on \emph{any} unbounded-treewidth instance class
(if it is treewidth-constructible).

The query $q_{\h}$, however, is not monotone.
We can alternatively show Theorem~\ref{thm:dichotomy} for a
MSO query which is \emph{monotone},
but not in FO: more specifically, we use a query in \rpqneq, the class of
\emph{conjunctive two-way regular path queries}
\cite{DBLP:conf/kr/CalvaneseGLV00,DBLP:journals/tcs/CalvaneseGV05} where we additionally allow
disequalities between variables.

We will show an analogue of Theorem~\ref{thm:dichotomy} in the setting of
\emph{tractable lineages} in Section~\ref{sec:meta}, which applies to \ucqneq,
an even weaker language.
We do not know whether Theorem~\ref{thm:dichotomy} itself can be shown with such
queries, or with a \emph{monotone} FO query. However, we know that
Theorem~\ref{thm:dichotomy} could not be shown with a query closed under
homomorphism; this is implied by Proposition~\ref{prp:nodichohomom}.

\paragraph*{Providing valuations with the instances}
When we fix the instance family $\I$,
the probability valuation is not prescribed as part of the family, but can be
freely chosen. If the instances of $\I$ were provided with
their probability valuations, or if probability valuations were forced to be
$1/2$, then it is unlikely that an equivalent to
Theorem~\ref{thm:dichotomy} would hold.

Indeed, fix \emph{any} query~$q$ such that, given any instance $I$, it is in \#P to count how many
subinstances of~$I$ satisfy $q$; e.g., let $q$ be a CQ. Consider
a family $\I$ of
instances \emph{with valuations} such that there is only one instance in~$\I$ per
encoding length: e.g., take the class of $R$-grids with probability $1/2$ on each
edge, for some binary relation $R$. Consider the problem, given the \emph{length}
of the encoding of an instance~$I$ (written in unary), of computing how many
subinstances of~$I$ satisfy~$q$. 
This problem is in the class
$\text{\#P}_1$~\cite{valiant1979complexity}. Hence, the probability computation problem for
$q$ on $\I$ is in $\fpspo$: rewrite the encoding of the input instance~$I$
to a word of the same length in a unary alphabet, use the $\text{\#P}_1$-oracle to compute
the number of subinstances, and normalize the result by dividing by the number of
possible worlds of~$I$. 

It thus
seems unlikely that probabilistic evaluation of $q$ on $\calI$ with its
valuations is \#P-hard, so that our dichotomy result probably does not adapt if
input instance families are provided with their valuations.

\section{Non-Probabilistic Evaluation}\label{sec:consequences}
\begin{table}
\caption{Summary of results for 
  non-probabilistic
  query evaluation: if a graph class has
  unbounded treewidth in \emph{some} sense, is closed under \emph{some} operations, and
  has (in data complexity) tractable
  model checking in \emph{some} sense for
\emph{some} logic, then \emph{some} complexity assumption is violated}
\label{tab:related}
\begin{minipage}{\linewidth}
  \scriptsize
    \renewcommand{\tabcolsep}{3pt}
    \begin{tabularx}{\linewidth}{llllXr@{~~}l}
\toprule
{\textbf{Logic}} &
{\textbf{Unboundedness}} &
{\textbf{Closure}} &
\textbf{Tractability}
 &
\textbf{Consequence} &
\multicolumn{2}{l}{\bfseries Source} \\
\midrule

MSO &
unbounded &
subgraph &
PTIME &
no violation: holds for some classes &
\cite{makowsky2003tree} & Prop.~32 \\

MSO$_2$ &
unbounded
&
subgraph &
PTIME &
no violation: holds for some classes &
\cite{kreutzer2010lower} & (remark)
\\

\midrule

$\exists$MSO &
unbounded &
topolog.\ minors &
PTIME &
P = NP &
\cite{makowsky2003tree} & Thm.~11 \\

MSO$_2$ &
strongly unb.\ polylog.

&
subgraph &
PTIME &
$\text{PH} \subseteq
  \text{DTIME}(2^{o(n)})$ &
\cite{kreutzer2010lower} & Thm.~1.2 \\

MSO &
densely unb.\ polylog.
&
subgr., vert.\ lab. &
quasi-poly &
$\text{PH} \subseteq
    \text{DTIME}(2^{o(n)})/\text{\textsc{subEXP}}$
    &
\cite{ganian2014lower} & Thm.~5.5 \\

\midrule

MSO &
unb., treewidth-constr. &
subgraph &
PTIME &
PH $\subseteq$ RP &
\bfseries here & \bfseries Thm.~\ref{thm:dichotomynp} \\

MSO &
densely unb.\ polylog.&
subgraph &
quasi-poly&
$\text{PH} \subseteq
\text{DTIME}(2^{o(n)})/\text{\textsc{subEXP}}$&
\bfseries here & \bfseries Thm.~\ref{thm:ganianvariant} \\

\bottomrule
\end{tabularx}
\end{minipage}

\end{table}

Theorem~\ref{thm:dichotomy} in Section~\ref{sec:probability} uses the
recent technology of~\cite{chekuri2014polynomial} that shows polynomial bounds for the grid
minor theorem of~\cite{robertson1986graph5}. These improved bounds also yield new results in the
non-probabilistic setting.
We accordingly study in this section the problem of 
\emph{non-probabilistic} query evaluation, again defined in terms of data
complexity:

\begin{definition}
  \label{def:msoeval}
  The \emph{evaluation problem} (or \emph{model-checking problem}), for a fixed
  query $q$ on an instance family $\I$, is as follows: given an
  instance $I \in \I$, decide whether $I \models q$.
\end{definition}

Observe that the probability evaluation problem in Section~\ref{sec:probability}
allowed the valuation to set edges to have
probability~$0$. We could thus restrict to any subinstance of an
instance in the class~$\calI$. In other words, the freedom to choose valuations
in probability evaluation gave us at least the possibility of
choosing subinstances for non-probabilistic query evaluation.
This is why we will study in this section the non-probabilistic query
evaluation problem on instance
classes~$\calI$ which are \emph{closed under taking subinstances} (or
\emph{subinstance-closed}), namely, for any
$I \in \calI$ and $I' \subseteq I$, we have $I' \in \calI$.

As before, we will prove dichotomy results for this problem on unbounded-treewidth
instance families, though we will use an MSO query rather than an FO
query. We give two phrasings of our results. The first one, in
Section~\ref{sec:hardness}, still requires
treewidth-constructibility, and shows hardness for every level of the polynomial
hierarchy, again under RP reductions.
The second phrasing, in Section~\ref{sec:alternate}, is inspired by the
results of~\cite{ganian2014lower}, which it generalizes: it relies on complexity assumptions
(namely, the non-uniform exponential time hypothesis) but works with a weaker
notion of constructibility, namely, it requires treewidth to be strongly
unbounded poly-logarithmically.

Last, we study in
Section~\ref{sec:counting} the problem of \emph{match
counting} in the non-probabilistic setting, for which no analogous results seemed to exist.

As in Section~\ref{sec:probability}, we restrict to signatures
of arity 2.

\subsection{Hardness Formulation}\label{sec:hardness}

Our first dichotomy result for non-probabilistic MSO query evaluation is as
follows; it is phrased using the notion of
treewidth-constructibility. In this result, $\sp_i$ denotes the complexity class
at the $i$-th existential level of the polynomial hierarchy.

\begin{theorem}\label{thm:dichotomynp}
  Let $\sigma$ be an arbitrary arity-2 signature.
  Let $\I$ be a
  class of 
  $\sigma$-instances
  which is treewidth-constructible and subinstance-closed.
  The following dichotomy holds:
  \begin{itemize}
  \setlength\itemsep{0pt}
    \item If there exists $k\in\mathbb{N}$ such that $\tw(I) \leq k$ for every $I\in\I$,
      then for every MSO
      query $q$, the evaluation problem for
      $q$ on~$\I$ is solvable in linear time.
    \item Otherwise, for each $i \in \NN$, there is an MSO query $q^i_{\h}$ (depending only on $\sigma$,
      not on $\I$) such that
      the evaluation problem for
      $q^i_{\h}$ on $\I$ is $\sp_i$-hard under RP reductions.
  \end{itemize}
\end{theorem}

The upper bound is by Courcelle's
results~\cite{courcelle1990monadic,flum2002query},
so our contribution is the hardness part, which we now sketch.

The
main thing to
change relative to the proof of
Theorem~\ref{thm:dichotomy} is the hard problems from which we reduce.
We use hard problems on
planar $\{1,3\}$-regular graphs, which we obtain from the \emph{alternating coloring
problem} as~\cite{ganian2014lower,ganian2010there}, restricted
to such graphs using techniques shown there, plus an additional construction to
remove vertex labellings. Here is our formal claim about the existence of such
hard
problems:

\begin{lemma}\label{lem:hardph}
  For any $i\in\mathbb{N}$, there exists an MSO formula $\psi_i$ on the
  signature of graphs such
  that the evaluation of $\psi_i$ on planar $\{1,3\}$-regular graphs is
  $\sp_i$-hard. Moreover, for any such graph~$G$, we have $G\models\psi_i$ iff
  $G' \models \psi_i$ for any subdivision $G'$ of~$G$.
\end{lemma}

The rest of the proof of Theorem~\ref{thm:dichotomynp} proceeds similarly
as that of Theorem~\ref{thm:dichotomy}.

\paragraph*{Hypotheses}

Theorem~\ref{thm:dichotomynp} relies crucially on the class~$\calI$ being
\emph{sub\-instance-closed}. Otherwise, considering the class $\calI$ of cliques
of a single binary relation $E$, this class is clearly unbounded-treewidth and
treewidth-constructible, yet it has bounded clique-width so MSO query evaluation
has linear data complexity on this class~\cite{courcelle2000linear}.

Further, the hypothesis of \emph{treewidth-constructibility} is also crucial. Without
this assumption, Proposition~32 of~\cite{makowsky2003tree} shows the existence of graph families
of unbounded treewidth which are subinstance-closed yet for which MSO
query evaluation is in PTIME.

\subsection{Alternate Formulation}\label{sec:alternate}

We now give an alternative phrasing of Theorem~\ref{thm:dichotomynp} which
connects it to the existing results of~\cite{kreutzer2010lower,ganian2014lower}.
Table~\ref{tab:related} tersely summarizes their results in comparison to our
own results and other related results.
As
\cite{kreutzer2010lower,ganian2014lower}
are phrased in
terms of graphs, and not arbitrary arity-2 relational
instances, we do so as well in this subsection. Before stating our result, we
summarize these earlier works to explain how our work relates to them.

\cite{kreutzer2010lower,ganian2014lower} 
show the intractability of MSO on any
subgraph-closed unbounded-treewidth families of graphs, under
finer notions than
our \emph{treewidth-constructibility}. Kreutzer and Tazari~\cite{kreutzer2010lower} proposed the notion of families of graphs
with treewidth \emph{strongly
unbounded poly-logarithmically} and showed that MSO$_2$
(MSO with quantifications over both vertex- and edge-sets) over any such
graph families is not fixed-parameter tractable in a strong sense (it is
not in XP), unless the exponential-time hypothesis (ETH) fails.
Ganian et al.~\cite{ganian2014lower} proved a related result, introducing
the weaker notion of \emph{densely unbounded poly-logarithmically}
but requiring graph families to be closed under \emph{vertex
relabeling}; in such a setting, Theorem~4.1 of~\cite{ganian2014lower} shows
that MSO (with vertex labels) cannot be fixed-parameter quasi-polynomial unless the
\emph{non-uniform} exponential-time hypothesis fails.

These two results of~\cite{kreutzer2010lower} and~\cite{ganian2014lower} are incomparable:
\cite{kreutzer2010lower} requires a stronger unboundedness notion
(strongly unbounded vs densely unbounded) and a stronger query language
(MSO$_2$ vs MSO), but it does not require vertex relabeling, and makes
a weaker complexity theory assumption (ETH vs
non-uniform ETH).
See the Introduction of~\cite{ganian2014lower} for a detailed comparison.

Our Theorem~\ref{thm:dichotomynp} uses MSO and no vertex labeling, but it requires \emph{treewidth-constructibility}, which is
stronger than
densely/strongly poly-logarithmic unboundedness: strongly
unboundedness only requires constructibility in $o(2^n)$ and densely
unboundedness does not require constructibility at all.
The advantage of treewidth-constructibility is that we were able to show
\emph{hardness} of our problem (under RP
reductions), without making \emph{any} complexity assumptions.
However, if we make the same complexity-theoretic hypotheses
as~\cite{ganian2014lower}, we now show that we can phrase our results in a
similar way to theirs, and thus strengthen them.

We accordingly recall the notion of densely poly-logarithmic
unboundedness, i.e., Definition~3.3 of~\cite{ganian2014lower}:
\begin{definition}
  A graph class~$\G$ has treewidth \emph{densely unbounded poly-logarithmically}
  if for all $c>1$, for all $m\in\NN$, there exists a graph $G\in\G$ such that
  $\tw(G)\geq m$ and $|V(G)|<O(2^{m^{1/c}})$.
\end{definition}

\begin{toappendix}
  In this appendix, we give the complete proof of
  Theorem~\ref{thm:ganianvariant}. The complete proofs of the other results of
  this section are given in Chapter~6 of~\cite{thesis}. We repeat
  that~\cite{thesis} also gives a proof of Theorem~\ref{thm:ganianvariant}, but
  it is in fact erroneous: specifically, the advice string defined in that proof
  depends on the input and not just on the length of the input.  This is why we
  propose a corrected proof below. Let us repeat the claim of
  Theorem~\ref{thm:ganianvariant}:
\end{toappendix}

We now state our intractability result on densely unbounded
poly-logarithmically graph classes. It is identical to Theorem~5.5 of~\cite{ganian2014lower} but applies to arbitrary MSO
formulae, without a need for vertex relabeling: in the result, $\mathrm{PH}$ denotes
the polynomial hierarchy. 
This result answers
Conjecture~8.3 of~\cite{grohe2008logic} (as we pointed out in the Introduction).
\begin{theoremrep}\label{thm:ganianvariant}
  Unless $\mathrm{PH}\subseteq\mathrm{DTIME}(2^{o(n)})/\textsc{SubEXP}$,
  there is no graph class $\G$ satisfying all three properties:
  \begin{enumerate}[label=\alph*)]
  \setlength\itemsep{0pt}
  \item $\G$ is closed under taking subgraphs;
  \item the treewidth of $\G$ is densely unbounded poly-logarith\-mi\-cally;
  \item the evaluation problem for any MSO query $q$ on $\G$ is
    quasi-polynomial, i.e., in time $O(n^{\log^d n\times
    f(|q|)})$ for $n=|V(G)|$, an arbitrary constant $d\geq 1$, and some
    computable function~$f$.
  \end{enumerate}
\end{theoremrep}

The proof technique is similar to that of~\cite{ganian2014lower} up to using the newer results
of~\cite{chekuri2014polynomial}. It is immediate that an analogous result holds for
probability query evaluation, 
as standard query evaluation obviously reduces to it
(take the probability valuation giving probability~$1$ to each fact).

\begin{toappendix}
  Recall that, in the statement of this result, $\mathrm{PH}$ refers to the polynomial
hierarchy, and $\textsc{SubEXP}$ refers to subexponential advice strings, i.e., the result
holds unless all problems in the polynomial-time hierarchy can be solved in
subexponential time with subexponential advice.

Informally, the advice that we will give for each input size
consists of a sufficiently large wall graph together with a topological
embedding of this wall graph in a graph of the family~$\calG$. Thanks to this
advice, we can find a topological embedding of arbitrary degree-3 planar graphs
inside the family and use it for our reduction.

  Our revised proof relies on Lemmas~\ref{lem:extract}
  and~\ref{lem:hardph}, which are proved in~\cite{thesis}, respectively as
  Lemma~6.1.4 and Lemma~6.2.2.
It also relies on the following fact, which is intuitively necessary
to ensure that arbitrary degree-3 planar graphs can be efficiently embedded into
walls as part of a polynomial-time reduction:

\begin{proposition}
  \label{prp:embed}
  Given as input a degree-3 planar graph~$G$ with~$n$ vertices, we can compute in polynomial time in~$G$
  a wall graph $W$ of size $O(n^4)$ by $O(n^4)$ and a topological embedding of~$G$ in~$W$.
\end{proposition}

  In this proposition, \emph{wall graphs}~\cite{dragan2011spanners} are a
  specific family of degree-3 planar graphs having high treewidth.
Proposition~\ref{prp:embed} follows from linear-time algorithms to topologically embed
planar graphs with maximal degree~$4$ into grids~\cite{tamassia1989planar}; we
explain this in more detail in a note~\cite{amarilli2023degree}
which also shows how to derive the result
from the linear-time algorithms
of~\cite{schnyder1990embedding,chrobak1995linear} to draw planar graphs
on a grid with integer coordinates.

We are now ready to prove the result:

\begin{proof}[Proof of Theorem~\ref{thm:ganianvariant}]
  Let us assume that there exists a graph class $\calG$ satisfying all three
  properties.
  Let $\calP$ be a problem in the polynomial hierarchy $\mathrm{PH}$, say in $\sp_i$,
  and let us show that we can solve $\calP$ in subexponential time with
  subexponential advice. Specifically, for any $0 < \epsilon < 1$, let us present an
  algorithm solving $\calP$ on instances of size $m$ with running time
  $O\left(2^{m^\epsilon)}\right)$ and advice of size
  $O\left(2^{m^\epsilon}\right)$.

  Let $F$ be a polynomial-time reduction from $\calP$ to
  the problem~$\psi_i$ of Lemma~\ref{lem:hardph}, which is hard for~$\sp_i$. 
  Let $l$ be the degree of this polynomial-time reduction~$F$; we can take $F$
  to assume without loss of generality that $l \geq 1$.
  Let $\alpha$ be the constant~$c$ given by Lemma~\ref{lem:extract}; again we can
  assume without loss of generality that $\alpha \geq 1$.
  Let $d$ be given by our assumption on~$\calG$ in the third point of the
  statement of Theorem~\ref{thm:ganianvariant}.
  Let us accordingly fix $c \defeq \frac{\epsilon}{8 \cdot 10 \cdot l \cdot
  \alpha \cdot (d+2)}$.

For each size~$m \in \NN$ of an instance of the problem $\calP$, let us define
  the advice for length~$m$. To do this, we use the fact that $\calG$ has
  treewidth densely unbounded poly-logarithmically. As $c < 1$, 
  we know that $\calG$
  contains a graph
  $G_m \in \calG$
  whose treewidth is
  $\geq m^{10 \cdot l \cdot \alpha}$
  and whose number of
vertices is in
  $O\left(2^{(m^{10 \cdot l \cdot \alpha})^{c}}\right)$, that is, in 
  $O\left(2^{m^{\epsilon/8(d+2)}}\right)$. Note that, as we are defining advice, the complexity of
  computing the graph~$G_m$ is irrelevant.

  Let us now consider the wall graph $W_l$ of size $m^{5l}$ by $m^{5l}$: it has exactly $m^{10l}$ vertices, and
  by Lemma~\ref{lem:extract} it is a topological minor of~$G_m$, so there is an
  embedding $(f_m, g_m)$ of~$W_l$ in~$G_m$; again the complexity
  of computing this embedding is irrelevant. This allows us to define the
  advice function: it maps $m$ to the triple $(G_m, f_m, g_m)$. The size of
  the advice as a function of~$m$ can be bounded using the fact that $G_m$ has
  $O\left(2^{m^{\epsilon/8 (d+2)}}\right)$ vertices,
  hence $O\left(2^{m^{\epsilon/4(d+2)}}\right)$ edges, and so the functions
  $f_m$ and $g_m$ can be represented with size
  $O\left(2^{m^{\epsilon/2(d+2)}}\right)$.
  Thus, the advice for input size $m$ is indeed in $O\left(2^{m^\epsilon}\right)$.

Now, let us explain how to decide the problem~$\calP$ on an input $x$ of size~$m$, 
  using the advice string $(G_m, f_m, g_m)$.
  We first run the polynomial-time reduction~$F$ on the input $x$
  to obtain a degree-3 planar graph $H \defeq F(x)$: 
  as $l$ is the degree of~$F$, the number of vertices and edges of~$H$ is in
  $O(m^l)$, and the answer of $\calP$ on~$x$ can be obtained by evaluating the
  MSO formula $\psi_i$ on~$H$. We then use
Proposition~\ref{prp:embed}: we can compute in polynomial time in~$H$ an embedding
of~$H$ into a wall graph of size $O(m^{4l})$ by $O(m^{4l})$. Now, for 
large enough~$m$, this is less than $m^{5l}$ by $m^{5l}$. As we can hardcode the
answer of the problem on small inputs, we assume that $m$ is indeed sufficiently
  large. Thus, we have obtained in polynomial time a topological embedding 
of~$H$ into the wall graph $W_l$ of size~$m^{5l}$ by $m^{5l}$.

  Combining this embedding with the embeddings $f_m$ and $g_m$ in the advice
  into the graph~$G_m$ also given in the advice,
  we obtain a topological embedding of~$H$ in~$G_m$. Thus, letting
  $G_m'$ be the image of this embedding, we know that $G_m'$ is also a graph
  of~$\calG$ (because $\calG$ is closed under taking subgraphs), and it is a
  subdivision of the input graph~$H$. Further, as
  $\psi_i$ is insensitive to subdivisions (as is ensured by
  Lemma~\ref{lem:hardph}), we know that we can evaluate $\psi_i$ on~$H$ by
  evaluating it on~$G_m'$. Now, as $G_m'$ is in $\calG$, by the third hypothesis
  on~$\calG$ in the statement of Theorem~\ref{thm:ganianvariant}, we know
  that, for some constant $K \defeq f(|\psi_i|)$, the evaluation problem
  for $\psi_i$ on~$G_m'$ runs in time $O\left(n^{\log^d n
  \times K}\right)$ for $n$ the number of vertices of~$G_m'$, i.e., in time
  $2^{O(\log^{d+1} n)}$. Now, we know that $n$ is in $O\left(2^{m^{\epsilon/8
  (d+2)}}\right)$, so the running time is in
  $2^{O\left(m^{\epsilon(d+1)/(8(d+2))}\right)}$. As $(d+1)/(8(d+2)) < 1$, the
  running time is indeed in $O\left(2^{m^\epsilon}\right)$.

  Thus, we have solved the problem $\calP$ with subexponential advice and
  subexponential running time. This violates the hypothesis, and concludes the proof.
\end{proof}

\end{toappendix}

\begin{table}
  \caption{Bounds for lineage representations (including intractable ones)}
  \label{tab:provenance}
  {
    \scriptsize
    \renewcommand{\tabcolsep}{5pt}
  \begin{tabularx}{\linewidth}{lllllrl}
    \multicolumn{7}{c}{\textbf{\footnotesize Upper bounds
(Section~\ref{sec:lineages}): computation bounds imply size bounds}}\\
    \toprule
    {\bfseries Instance} & {\bfseries Queries} & {\bfseries Representation} &
    {\bfseries Note} & 
    {\bfseries Time} &
\multicolumn{2}{l}{\bfseries Source} \\
    \midrule
    bounded-pw & MSO & OBDD &  $O(1)$ width &$O(n)$ &
    \bfseries here& \bfseries Thm.\ \ref{thm:pathwidth}\\
    bounded-pw & (monotone) MSO& (monotone) circuit & bounded-pw&$O(n)$ &
    \bfseries here&\bfseries Prop.~\ref{prp:makecircuitpath}\\
     bounded-tw & MSO & OBDD && $O(\text{Poly}(n))$ &
    \bfseries here&\bfseries Thm.\ \ref{thm:getobdd}\\
    bounded-tw & (monotone) MSO&  (monotone) circuit & bounded-tw&$O(n)$
    & \cite{amarilli2015provenance} & Thm. 4.2\\
    bounded-tw & MSO & d-DNNF && $O(n)$ & \bfseries here & \bfseries Thm. \ref{thm:makeddnnf}\\
    any & inversion-free UCQ & OBDD & $O(1)$ width & $O(\text{Poly}(n))$ &
    \cite{jha2013knowledge} & Prop. 5 \\
    any & positive relational algebra & monotone formula && $O(\text{Poly}(n))$ &
    \cite{DBLP:journals/jacm/ImielinskiL84}&Thm. 7.1\\
    any & Datalog & monotone circuit &&  $O(\text{Poly}(n))$ &
    \cite{deutch2014circuits} & Thm. 2 \\
    \bottomrule
    ~\\
    \multicolumn{7}{c}{\textbf{\footnotesize Lower bounds (Section~\ref{sec:formulae})}}\\
    \toprule
    {\bfseries Instance} & {\bfseries Queries} & {\bfseries Representation} &&
    {\bfseries Size} & {\bfseries Source}\\
    \midrule
    tree & CQ$^{\neq}$ & formula && $\Omega(n \log \log n)$ &
    \bfseries here &\bfseries Prop.~\ref{prp:provlowercq}\\
    tree & CQ$^{\neq}$ & monotone formula && $\Omega(n \log n)$ &
    \bfseries here &\bfseries Prop.~\ref{prp:provlowercqpos}\\
    tree & MSO & formula && $\Omega(n^2)$ & \bfseries here &\bfseries Prop.~\ref{prp:provlowertree}\\
    any & Datalog & monotone formula && $n^{\Omega(\log n)}$ &
    \cite{deutch2014circuits} & Thm. 1\\
    \bottomrule
  \end{tabularx}
  }
\end{table}

\subsection{Match Counting}\label{sec:counting}
We conclude this section by moving to the problem of \emph{match
counting}, i.e., counting how many assignments satisfy a \emph{non-Boolean
MSO formula}.
Match counting should not be confused with \emph{model counting}
(counting how many subinstances satisfy a Boolean formula) which
is closely related\footnote{The number of models of a
query~$q$ in an instance $I$ of size $n$ is $2^n$ multiplied by the
probability of~$q$ under the probability valuation of $I$ that gives probability
$1/2$ to each fact of~$I$.} to probability evaluation.

To our knowledge, no dichotomy-like result on match counting for MSO queries was known.
This section shows such a result; as in Section~\ref{sec:hardness}, we assume
treewidth-constructibility, closure
under subinstances, and arity-2 signatures.

We define the match counting problem as
follows:

\begin{definition}
  The \emph{counting problem} for an
  MSO formula $q(\mathbf{X})$ (with free second-order variables)
  on an instance family $\I$ is the problem, given an
  instance $I \in \I$, of counting how many vectors $\mathbf{A}$ of domain
  subsets are such that $I$ satisfies $q(\mathbf{A})$.

  The restriction to free second-order variables is without loss of generality,
  as
  free first-order variables can be rewritten to free second-order ones,
  asserting in the formula that they must be interpreted as singletons.
\end{definition}

We show the following dichotomy result:

\begin{theorem}\label{thm:dichotomynpc}
  Let $\sigma$ be an arbitrary arity-2 signature.
  Let $\I$ be a
  subinstance-closed and treewidth-constructible class of $\sigma$-instances.
  The following dichotomy holds:
  \begin{itemize}
  \setlength\itemsep{0pt}
    \item If there exists $k\in\mathbb{N}$ such that 
      $\tw(I)\leq k$
      for every $I\in\I$, then for every MSO query
      $q(\mathbf{X})$ with free second-order variables, the counting problem for
      $q$ on~$\I$ is solvable in ra-linear time.
    \item Otherwise, there is an MSO query $q'_{\h}(X)$ (depending only on $\sigma$,
      not on $\I$) with one free second order variable such that
      the counting problem for
      $q_{\h}'$ on~$\I$ is $\fpsp$-complete under RP reductions.
  \end{itemize}
\end{theorem}

The first claim is shown in~\cite{arnborg1991easy}.
The proof of the second claim proceeds as for
Theorem~\ref{thm:dichotomy}. We reduce from the problem of counting
Hamiltonian cycles in planar 3-regular graphs~$G$, which is
\#P-hard~\cite{LiskiewiczOT03}, and which we express in MSO on the incidence
graph of~$G$.

Unlike in Theorem~\ref{thm:dichotomy}, the query $q_\h'$ does not
have tractable model checking (as opposed to probability
evaluation). We do not know whether we can show a similar result with such a
tractable query.

\section{Lineage Upper Bounds}
\label{sec:lineages}
From \emph{probability evaluation} in Section~\ref{sec:probability} (and its
non-probabilistic variants in Section~\ref{sec:consequences}),
we now turn to our second problem: the study of \emph{tractable lineage representations}.

Indeed, a common way to achieve tractable probability evaluation
is to represent the \emph{lineage} of queries on input instances
in a \emph{tractable} formalism~\cite{jha2013knowledge}.
This section shows how the tractability of MSO probability evaluation on
bounded-treewidth instances can be explained via lineages: 
Table~\ref{tab:provenance} (upper part) summarizes
the upper bounds that
we prove.

Intuitively, the \emph{lineage} of a query on an instance describes how the query depends on
the facts of the instance. Formally:

\begin{definition}
  The \emph{lineage} of a query $q$ on an instance~$I$ is a Boolean
  function $\phi$ whose variables are the facts of $I$, such that, for any $I'
  \subseteq I$, $I' \models q$ iff the corresponding valuation makes $\phi$
  true.
  If $q$ is monotone, then $\phi$ is a monotone Boolean function, in
  which case it can equivalently be called the
  $\posbool{X}$-provenance~\cite{green2007provenance} of $q$ on~$I$.
\end{definition}

Lineages are related to probability evaluation, because evaluating the
probability of query~$q$ under a probability valuation $\pi$ of instance~$I$
amounts to evaluating the probability of the lineage $\phi$, under the
corresponding
probability valuation on variables. Thus, if we can represent
$\phi$ in a formalism that enjoys tractable probability
computation, then we can tractably evaluate the probability $\pi(q, I)$ of~$q$ on~$I$.

In this section, we show that MSO queries on bounded-treewidth instances admit
tractable lineage representations in many common formalisms: they have
linear-size bounded-treewidth
Boolean circuits (as shown
in~\cite{amarilli2015provenance}), but also have
polynomial-size OBDDs~\cite{bryant1992symbolic,olteanu2008using} (with a
stronger claim for \emph{bounded-pathwidth}), as
well as
linear-size d-DNNFs~\cite{darwiche2001tractable}. Further, as we show, all these lineage
representations
can be efficiently computed.
Note that in all these results, as in~\cite{amarilli2015provenance},
tractability only refers to \emph{data complexity}, with large
constant factors in the query and instance width: we leave to future work the
study of query and instance classes for which lineage computation
enjoys a lower \emph{combined complexity}.
 
After our results on tractable lineage computations in this
section,
we will investigate in the next section whether we can represent the lineage as \emph{Boolean formulae} (such as
read-once formulae~\cite{jha2013knowledge}), and we will show superlinear lower bounds
for them. Section~\ref{sec:meta} will then study in which sense
bounded-treewidth is
\emph{necessary} to obtain tractable lineages.

This section applies to
signatures of arbitrary arity.

\paragraph*{Bounded-treewidth circuits}
We first recall our results from~\cite{amarilli2015provenance} and introduce
a first representation of lineages: \emph{Boolean circuits}, called
\emph{provenance circuits} in~\cite{amarilli2015provenance} and
\emph{expression DAGs} in~\cite{jha2012tractability}:

\begin{definition}
  \label{def:circuit}
  A \emph{lineage circuit} for query $q$ and instance~$I$ is a Boolean
  circuit with input gates and with NOT, OR, and AND internal gates,
  whose inputs are the facts of the database, and which computes the lineage of $q$
  on $I$.
A \emph{monotone} lineage circuit has no NOT
  gate.
  The \emph{treewidth} and \emph{pathwidth} of a lineage circuit are
  that of the circuit's graph.
\end{definition}

As recalled in Theorem~\ref{thm:main}, \cite{amarilli2015provenance} showed that
we can compute (monotone) lineage circuits for (monotone) MSO queries on
bounded-treewidth instances in linear time. Further, these circuits themselves
have \emph{bounded-treewidth}, which is why probability evaluation is
tractable on them,
using message passing algorithms~\cite{lauritzen1988local}. Hence:

\begin{theorem}[(\cite{amarilli2015provenance}, Theorems~4.4 and~5.3)]
  \label{thm:makecircuit}
  For any fixed MSO query $q$ and constant $k \in \NN$, given an input instance
  $I$ of treewidth $\leq k$, we can compute in linear time a bounded-treewidth
  lineage circuit $C$ of $q$ on $I$.

  If $q$ is monotone then we can
  take $C$ to be monotone.
\end{theorem}

We study how to adapt this to other tractable lineage representations.

\paragraph*{OBDDs}
We start by defining \emph{OBDDs}, a common tractable representation of Boolean
functions~\cite{bryant1992symbolic,olteanu2008using}:

\begin{definition}
  An \emph{ordered binary decision diagram} (or \emph{OBDD}) is a rooted
  directed acyclic graph (DAG) whose leaves
  are labeled $0$ or $1$, and whose non-leaf nodes are labeled with a variable
  and have two outgoing edges labeled $0$ and $1$.
  We
  require that there exists a total order $\Pi$ on the variables such that, for every
  path from the root to a leaf, no variable occurs in two different internal
  nodes on the path, and the order in which the variables occur is compatible
  with~$\Pi$.

  An OBDD defines a Boolean function on its variables: each valuation is mapped
  to the value of the leaf reached from the root by following the path given
  by the valuation.

  The \emph{size} of an OBDD is its number of nodes, and its \emph{width} is the
  maximum number of nodes at every \emph{level}, where a level is the set of
  nodes reachable by enumerating all possible values of variables in a prefix
  of~$\Pi$.
\end{definition}

Probability evaluation for OBDDs is tractable~\cite{olteanu2008using}.
Our result is that we can compute
polynomial-size OBDDs for MSO queries on bounded-treewidth instances in PTIME:

\begin{theorem}
  \label{thm:getobdd}
  For any fixed MSO query $q$ and constant $k \in \NN$, there is $c
  \in \NN$ such that,
  given an input instance
  $I$ of treewidth $\leq k$, one can compute in time $O(\card{I}^c)$ an
  OBDD (of size $O(\card{I}^c)$)
  for the lineage of~$q$ on~$I$.
\end{theorem}

We show this using Corollary~2.14 of~\cite{jha2012tractability}:
any bounded-treewidth Boolean circuit can be represented by an equivalent OBDD of
polynomial width. We complete this result and show that the OBDD can also be
computed in polynomial time, which clearly implies Theorem~\ref{thm:getobdd}
(using Theorem~\ref{thm:makecircuit}):

\begin{lemmarep}\label{lem:makeobddtw}
  For any $k \in \NN$, there is $c \in \NN$ such that,
  given a Boolean circuit $C$ of treewidth $k$, 
  we can compute an equivalent OBDD in time~$O(\card{C}^c)$.
\end{lemmarep}

\begin{proofsketch}
  The variable order is that of~\cite{jha2012tractability} and can be
  constructed in PTIME. We then show that we can build the polynomial-size OBDD level by
  level, by testing the equivalence of partial valuations in PTIME. We do it
  using message passing, thanks to the tree decomposition of~$C$. As
  in~\cite{jha2012tractability}, $c$ is doubly exponential
  in~$k$.
\end{proofsketch}

\begin{proof}
  We rely on Corollary~2.14 of~\cite{jha2012tractability}: 
  there is a doubly exponential function~$f$ such that,
  for any $k \in \NN$,
  there is $c' \defeq f(k)\in \NN$ 
  such that, for any tree decomposition $T$ of width $\leq k$ of $C$,
  the OBDD $O$ obtained for a certain variable order $\Pi^R$ has
  width~$c'$.
  
  The order $\Pi^R$ on variables is defined following an
  in-order traversal of~$T$ where children are ordered by the number of variables in
  this subtree; clearly this quantity can be computed over the entire tree in
  PTIME, so the order $\Pi^R$ can be computed in PTIME. We show that we can
  construct the OBDD $O$ in PTIME as well, in a level-wise manner inspired by~\cite{jha2013knowledge}.

  Write $\Pi^R = X_1, \ldots, X_n$, and construct $O$ level-by-level in the
  following way. Assuming that we have constructed $O$ up to level $l-1$,
  create two children for each node at level $l-1$ (depending on the value of
  variable $X_l$), and then merge all such children $n$ and $n'$ that are
  \emph{equivalent}. To define this, call \emph{equivalent} two partial
  valuations $\nu$ and $\nu'$ of variables $X_1, \ldots, X_l$ if the Boolean
  function represented by $C$ on the other variables $X_{l+1}, \ldots, X_n$
  under $\nu$ is the same as under $\nu'$. Now, call $n$ and $n'$
  \emph{equivalent} if, for any partial valuation $\nu$ leading to $n$
  (represented by a path from the root of~$O$ to $n$) and any partial valuation
  $\nu'$ leading to $n'$, these two partial valuations of $X_1, \ldots, X_l$ are
  equivalent. As we will always ensure in the construction, any paths leading to
  the parent node of $n$ are equivalent partial valuations of $X_1, \ldots,
  X_{l-1}$, so $n$ and $n'$ are equivalent iff, picking any two valuations $\nu$
  for $n$ and $\nu'$ for $n'$ by following a path from the root to~$n$ and
  to~$n'$ respectively, $\nu$ and $\nu'$ are equivalent.

  Hence, it suffices to show that there is a function $g$ such that we can test
  in time $O(\card{C}^{g(k)})$ whether two partial valuations are equivalent.
  Indeed, we can then build $O$ in the indicated time, because the
  maximal number of node pairs to test at any level of the OBDD is
  $\leq (2\cdot \card{C}^{f(k)})^2$: we had at most $\card{C}^{f(k)}$ at the
  previous level, and each of them creates two children, before we merge the
  equivalent children. Hence, if we can test equivalence in the indicated time,
  then clearly we can construct $O$ in time $O(\card{C}^c)$ for $c \defeq 1 + 1 +
  2\cdot f(k) + g(k)$ (the first term accounts for the linear number of levels,
  and the second term accounts for the linear time required to find a partial
  valuation for a node).

  We thus show that the equivalence of partial valuations can be tested in time
  $O(\card{C}^{g(k)})$ for some function~$g$.
  Considering two partial valuations $\nu$ and $\nu'$ of the same set of
  variables $\calX$, let $C_{\nu}$ and $C_{\nu'}$ be the two circuits obtained from $C$
  by substituting the input gates for $\calX$ with constant gates according to
  $\nu$ and $\nu'$ respectively. Note that $C_{\nu}$ and $C_{\nu'}$ have the same
  set of input gates $\calX'$, formed precisely of the variables not in $\calX$.
  We rename the internal gates of~$C_{\nu'}$ so that the only gates shared
  between $C_{\nu}$ and $C_{\nu'}$ are the input gates $\calX'$.
  Now, $C'$ be the circuit obtained by taking the union of $C_{\nu}$ and
  $C_{\nu'}$ (on the same set of variables), and adding an output gate and a
  constant number of gates such that the output gate is true iff the output
  gates of $C_{\nu}$ and~$C_{\nu'}$ carry different values (this can be done
  with 5 additional gates in total). It is easy to see
  that there is a valuation of $\calX'$ that makes the circuit $C'$ evaluate to
  true iff the partial valuations $\nu$ and $\nu'$ are \emph{not} equivalent. Now,
  observe that we can immediately construct from $T$ a tree decomposition~$T''$ of
  width $\leq 2k + 5$ of $C'$. Indeed, it is obvious that $T$ is a tree
  decomposition of~$C_{\nu}$, and we can rename gates to obtain from $T$ a tree
  decomposition $T'$ of~$C_{\nu'}$, such that $T$ and~$T'$ both have the same
  width $k$ and the same skeleton. Now, construct $T''$ that has same skeleton as
  $T$ and $T'$ where each bag is the union of the corresponding bags of~$T$
  and~$T'$, adding the 5 intermediate gates to each bag.
  The result $T''$ clearly has width $\leq 2k +
  5$ and it is immediate that it is a tree decomposition of~$C'$.

  We can then use message-passing
  techniques~\cite{lauritzen1988local,huang1996inference} to
    determine in time exponential in
  $2k+5$ and polynomial in~$C'$ whether the bounded-treewidth circuit $C'$ has a
  satisfying assignment, from which we deduce whether $\nu$ and $\nu'$ are
  equivalent. For details, see, e.g., Theorem~D.2 of~\cite{amarilli2015provenance}.
\end{proof}

\paragraph*{Bounded-pathwidth}
We have explained the tractability of MSO probability evaluation on bounded-treewidth
instances, showing that we could compute bounded-treewidth lineage circuits for them.
We strengthen these results in the case of
\emph{bounded-pathwidth} instances, showing that we can compute
\emph{constant-width} OBDDs:

\begin{theorem}
  \label{thm:pathwidth}
  For any fixed MSO query $q$ and constant $k \in \NN$, given an input instance
  $I$ of pathwidth $\leq k$, one can compute in polynomial time an
  OBDD of constant width for the lineage of~$q$ on~$I$.
\end{theorem}

To prove the result, we first observe 
(adapting~\cite{amarilli2015provenance})
that we can compute \emph{bounded-pathwidth}
lineage circuits in linear time on bounded-pathwidth instances:

\begin{propositionrep}\label{prp:makecircuitpath}
  For any fixed $k \in \NN$ and (monotone) MSO query $q$, for any
  $\sigma$-instance~$I$ of \emph{pathwidth} $\leq k$, 
  we can construct a (monotone) lineage circuit $C$ of~$q$ on~$I$ in time
  $O(\card{I})$. The \emph{pathwidth} of~$C$ only depends on~$k$ and~$q$ (not
  on~$I$).
\end{propositionrep}

\begin{proof}
  Given a path decomposition of an instance~$I$, which is a tree decomposition
  with a linear tree, the resulting tree encoding~$E$ of~$I$
  (see \cite{amarilli2015provenance,amarilli2015provenanceb})
  is clearly also a
  linear tree.
  From the proof of Theorem~4.4
  of~\cite{amarilli2015provenanceb}, we observe that the lineage circuit that
  we construct has a tree decomposition which can be made to be a path
  decomposition in this case, because it follows the structure of~$E$. 
  Hence, the circuit $C$ has bounded pathwidth.
\end{proof}

By Corollary~2.13 of~\cite{jha2012tractability}, this implies the existence of a constant-width
OBDD representation, which we again show to be computable, proving
Theorem~\ref{thm:pathwidth}.

\begin{lemmarep}\label{lem:makeobddpw}
  For any $k \in \NN$,
    for any Boolean circuit $C$ of pathwidth $\leq k$,
      we can compute in polynomial time in~$C$
        an OBDD equivalent to $C$ whose width depends only on~$k$.
\end{lemmarep}

\begin{proof}
  As in the proof of Lemma~\ref{lem:makeobddtw},
  we can compute in PTIME the order $\Pi^R$ on variables,
  and we can compute the OBDD under this order in the same way.
  This uses the fact
  that a path decomposition of circuit~$C$ is in particular a tree decomposition
  of~$C$.
\end{proof}

\paragraph*{d-DNNFs}
We now turn to the more expressive tractable lineage formalism of
\emph{d-DNNFs}, introduced in~\cite{darwiche2001tractable}; we follow the
definitions of~\cite{jha2013knowledge}:

\begin{definition}
  A \emph{deterministic, decomposable negation normal form}
  (\emph{d-DNNF}) is a Boolean circuit $C$
  that satisfies the following conditions:
  \vspace{-.2cm}\begin{enumerate}
  \setlength\itemsep{0pt}
    \item Negation is only applied to input gates: the input of any NOT gate must
      always be an input gate.
    \item The inputs of AND-gates depend on disjoint sets of input gates.
      Formally, for any AND-gate $g$,
      for any two gates $g_1 \neq g_2$ which are inputs of~$g$,
      there is no input gate $g'$ which is reachable (as a possibly indirect input) from
      both $g_1$ and $g_2$.
    \item The inputs of OR-gates are mutually exclusive.
      Formally, for any OR-gate $g$,
      for any two gates $g_1 \neq g_2$ which are inputs of~$g$,
      there is no valuation of the inputs of~$C$ under which $g_1$ and~$g_2$ both evaluate
      to true.
  \end{enumerate}
\end{definition}

It is tractable to evaluate the probability of a
\mbox{d-DNNF}~\cite{darwiche2001tractable}, and
d-DNNFs capture the tractability of probability evaluation for
many safe queries (see~\cite{jha2013knowledge}). We show that it also explains
the ra-linearity of MSO probability evaluation on bounded-treewidth instances,
as we can construct \emph{linear} d-DNNFs for them:

\begin{theoremrep}\label{thm:makeddnnf}
  For any fixed MSO query $q$ and constant $k \in \NN$, given an input instance
  $I$ of treewidth $\leq k$, one can compute in time $O(\card{I})$ a d-DNNF
  representation of the lineage of~$q$ on~$I$.
\end{theoremrep}

\begin{proofsketch}
  Our construction
  in~\cite{amarilli2015provenance} applies a tree automaton
  for the query to an
  annotated tree encoding of the instance. This yields a bounded-treewidth
  circuit representation of the lineage, in linear time in the instance (but
  with a constant factor that is nonelementary in the query).
  We show that, if the automaton is deterministic, the circuit that we
  obtain is already a d-DNNF. The result follows, as one
  can always make a tree automaton deterministic~\cite{tata}, at the cost of an
  increased constant factor in the data complexity.
\end{proofsketch}

\begin{proof}
  We define a \emph{bottom-up deterministic tree automaton} on alphabet $\Gamma$
  (or $\Gamma$-bDTA) in the standard manner.
  We start by adapting the proof of Proposition~3.1
  of~\cite{amarilli2015provenanceb} to show
  the following result instead: a provenance d-DNNF of a
  \emph{deterministic} $\overline{\Gamma}$-bDTA $A$
  on a $\overline\Gamma$-tree $E$ can be constructed in time $O(\card{A} \cdot
  \card{E})$. We construct the circuit exactly as in the proof of
  Proposition~3.1 of~\cite{amarilli2015provenanceb} and show that it is a d-DNNF.

    First, observe that the only NOT gates that we use are the
      $g_n^{\lnot\i}$,
      which are NOT gates of the $g^{\i}_n$, which are input gates; so we only
      apply negation to leaf nodes.

    Second, we show that the sets of leaves reachable from the children of
      any AND gate are pairwise disjoint. The AND gates that we create and that
      have multiple inputs are:
      \begin{itemize}[leftmargin=9pt]
        \item The $g_n^{q_{\LC}, q_{\RC}}$, which are the AND of
          $g^{q_{\LC}}_{\LCf(n)}$ and $g^{q_{\LC}}_{\RCf(n)}$; now, $g^{q_{\LC}}_{\LCf(n)}$ only
          depends on the input gates $g^{\i}_{n'}$ for nodes $n'$ of the subtree
          of~$E$ rooted at $\LCf(n)$, and likewise $g^{q_{\LC}}_{\RCf(n)}$ only depends
          on input gates in the right subtree;
        \item The $g_n^{q_{\LC}, q_{\RC}, \i}$, which are the AND of $g_n^{q_{\LC},
          q_{\RC}}$ and $g^{\i}_n$; now, the $g_n^{q_{\LC}, q_{\RC}}$ do not depend
          on $g^{\i}_n$, only on input gates $g^{\i}_{n'}$ for $n'$ a strict
          descendant of $n$ in~$E$;
        \item The $g_n^{q_{\LC}, q_{\RC}, \lnot\i}$, which are the AND of $g_n^{q_{\LC},
          q_{\RC}}$ and $g^{\lnot\i}_n$; now, the $g_n^{q_{\LC}, q_{\RC}}$ do not depend
          on the sole input gate under $g^{\lnot\i}_n$, i.e.,
          $g^{\i}_n$, but only on input gates $g^{\\i}_{n'}$ for $n'$ a strict
          descendant of $n$ in~$E$.
      \end{itemize}
    Third, we show that the children of any OR gate are mutually
      exclusive. The OR gates that we create and that have multiple
      inputs are the following:
      \begin{itemize}[leftmargin=9pt]
        \item The $g^q_n$ when $n$ is a leaf node of~$E$, for which the claim is
          immediate, as the only two possible children are $g^{\i}_n$ and
          $g^{\neg \i}_n$ which are clearly mutually exclusive.
        \item The $g^q_n$ when $n$ is an internal node of~$E$, which are the OR of gates of the form
          $g^{q_{\LC}, q_{\RC},\i}_n$ or $g^{q_{\LC}, q_{\RC},\lnot\i}_n$ over
          several pairs $q_{\LC}, q_{\RC}$.
          
          To observe that these gates are
          mutually exclusive, remember that, for a valuation $\nu$ of the tree
          $E$, the gate $g^q_{n'}$ is true iff there is a run~$\rho$ of $A$ on
          the subtree of~$\nu(E)$ rooted at~$n'$ such that $\rho(n') = q$.
          However, as
          $A$ is deterministic, for each $n'$, there is at most one state~$q$
          for which this is possible.
          Hence, for any valuation $\nu'$ of the circuit $C$, for our
          node $n$, there is at most one $q_{\LC}'$ such that
          $g^{q_{\LC}'}_{\LCf(n)}$ is true under valuation $\nu'$, and only at
          most one $q_{\RC}'$ such that
          $g^{q_{\RC}'}_{\RCf(n)}$ is true under $\nu'$. Hence, by definition of the
          $g^{q_{\LC}, q_{\RC}}_n$, there is at most one of them which can
          be true under valuation $\nu'$, namely, $g^{q_{\LC}',
          q_{\RC}'}_n$, which also means that only the gate $g^{q_{\LC}',
          q_{\RC}',\i}_n$ and the gate $g^{q_{\LC}',
          q_{\RC}',\lnot\i}_n$ can be true under $\nu'$. But these two
          gates are clearly mutually exclusive (only one can evaluate to true,
          depending on the value of~$\nu(n)$), which proves the
          claim.
        \item The output gate $g_0$ which is the OR of gates of the form $g^q_r$ for $r$ the
          root node of $E$. Again, as $A$ is deterministic, for
          any valuation~$\nu'$ of~$C$, letting $\nu$ be the corresponding
          valuation of the $\overline{\Gamma}$-tree $E$, there is only one state $q'$
          such that $A$ has a run $\rho$ on~$E$ with $\rho(r) = q'$, so at most
          one state $q'$ such that $g^{q'}_r$ is true under $\nu$.
      \end{itemize}
    Hence, the circuit constructed in the proof of Proposition~3.1
    of~\cite{amarilli2015provenance} is a d-DNNF representation of the lineage
    of the automaton on the tree which has linear size.

    We now adapt the proof of Theorem~4.2
    of~\cite{amarilli2015provenanceb}. The
    theorem proceeds by constructing a bNTA for the
    query~$q$~\cite{courcelle1990monadic}
    on the alphabet $\Gamma_\sigma^k$ and modifying it to obtain a bNTA $A'$ on
    $\overline{\Gamma_\sigma^k}$. We now additionally convert $A'$ to a bDTA $A''$ on
    the same alphabet, which we can do using standard
    techniques~\cite{tata}. All of this is performed independently of the
    instance.

    Now, we conclude using the rest of the proof of Theorem~4.2
    of~\cite{amarilli2015provenance}. The resulting circuit $C'$
    is the result of (bijectively) renaming the input gates, and replacing some
    input gates by constant gates, on the circuit $C$ produced by
    Proposition~3.1 of~\cite{amarilli2015provenance}. However, by our previous
    observation, $C$ is actually a d-DNNF circuit, so $C'$ also is (up to
    evaluating negations of constant gates as constant gates). Hence, we
    have produced the desired d-DNNF, which by the statement of Theorem~4.2
    of~\cite{amarilli2015provenance} is
    of linear size and is computed in linear time.
\end{proof}

\section{Formula Lower Bounds}
\label{sec:formulae}
We have shown that MSO queries on bounded-treewidth instances have tractable
lineages, and even linear-sized ones (bounded-treewidth circuits and d-DNNFs).
All
lineage representations that we have studied, however, are based on DAGs (circuits
or OBDDs).
In this section, we study whether we could obtain linear lineage
representations as \emph{Boolean formulae}, e.g., as \emph{read-once
formulae}~\cite{jha2013knowledge}.

This question is natural because
existing work on probabilistic data seldom represents query lineages as
Boolean circuits: they tend to use Boolean formulae, or other representations such as
OBDDs, FBDDs and d-DNNFs~\cite{jha2013knowledge}. Further, 
most prior works on provenance  focus on formula representations of provenance
(with the notable exception of~\cite{deutch2014circuits}, see below).

Intuitively, an important difference between formula and circuit representations
is that circuits can share common subformulae whereas formulae cannot. We
show in this section that this difference matters: formula-based representations of
the lineage of MSO queries on bounded-treewidth instances \emph{cannot} be
linear in general, because of superlinear lower bounds.
More specifically, we show that formula-based representations are
\emph{superlinear} even for some
\cqneq queries, and that they are \emph{quadratic} for some MSO queries.

A similar
result was already known for lineage representations
(\cite{deutch2014circuits}, Theorem~1), which showed that circuit representations of
provenance can be more concise than formulae; but this result applies to
\emph{arbitrary} instances, not bounded-treewidth ones. Hence, this section also
sheds additional light on the compactness gain offered by the recent \emph{circuit}
representations of provenance in~\cite{deutch2014circuits}.

Our results in this section rely on
classical lower bounds on the size of
formulae expressing certain Boolean
functions~\cite{wegener1987complexity}. They apply to signatures of arbitrary arity.
We summarize them in the lower part
of Table~\ref{tab:provenance}.

\paragraph*{\cqneq queries}
We first show a mild conciseness gap for the comparatively simple language of
\cqneq.
Formally, we exhibit
a \cqneq query whose lineage cannot be represented by a
linear-size formula, even on treelike instance families. By contrast,
from the previous section,
we know that its lineage has linear-size circuit representations.

\begin{propositionrep}\label{prp:provlowercq}
  There is a
  \cqneq query~$q$, and a family~$\I$ 
  of relational instances with
  treewidth~$0$, such that,
  for any $I \in \I$, for any Boolean formula $\psi$ capturing
  the lineage of~$q$ on~$I$,
  we have
  $\card{\psi} = \Omega(\card{I} \log \log
  \card{I})$.
\end{propositionrep}

\begin{proofsketch}
  We show we can express the threshold function over $n$ variables, for
  which \cite{wegener1987complexity} gives a lower bound on the formula
  size.
\end{proofsketch}

\begin{proof}
  Consider the signature with a single unary predicate $R$, and consider the
  \cqneq $q: \exists x y \, R(x) \wedge R(y) \wedge x \neq y$. Consider
  the family of instances~$\I$ defined as
  $\{R(a_1), \ldots, R(a_n)\}$ for all $n \in \NN$.
  Clearly, for any $I \in \I$, the lineage of~$q$ on~$I$ is the threshold
  function checking whether at least two of its inputs are true. 
 
  By Chapter~8, Theorem~5.2 of~\cite{wegener1987complexity}, 
  any formula using
  $\land$, $\lor$, $\lnot$ expressing
  the threshold function over $n$ variables has size 
  $\Omega(n\log\log n)$, the desired bound.
\end{proof}

As \cqneq{} queries are monotone, we can also ask for \emph{monotone}
lineage representations. From the previous section, we still have
linear representations of the lineage as a \emph{monotone} circuit. By contrast,
if we restrict to \emph{monotone} Boolean formulae, we obtain an improved lower bound:

\begin{propositionrep}\label{prp:provlowercqpos}
  There is a
  \cqneq query~$q$,
  and a family~$\I$
  of relational instances with
  treewidth~$0$, such that,
  for any $I \in \I$, for any \emph{monotone} Boolean formula $\psi$ capturing
  the lineage of~$q$ on~$I$, we have
  $\card{\psi} = \Omega(\card{I} \log 
  \card{I})$.
\end{propositionrep}

\begin{proof}
  We use the same proof as for Proposition~\ref{prp:provlowercq} but
  relying on~\cite{hansel1964nombre} (also
  Chapter~8, Theorem~1.2 of~\cite{wegener1987complexity}), which shows that, for $\psi$ built over the
  monotone basis $\land$ and $\lor$, we have
  $\card{\psi}=\Omega(n\log n)$.
\end{proof}

\paragraph*{MSO queries}
For the more expressive language of MSO queries, we show a wider gap: there is a
query for which formula-based lineage representations must be \emph{quadratic},
whereas we know that there are linear circuit representations of the lineage:

\begin{propositionrep}\label{prp:provlowertree}
  There is 
  an MSO query~$q$,
  and a family~$\I$ 
  of relational instances with
  treewidth~$1$, such that,
  for any $I \in \I$, for any Boolean formula $\psi$ capturing the lineage
  of~$q$ on~$I$,
  we have
  $\card{\psi} = \Omega(\card{I}^2)$.
\end{propositionrep}

\begin{proofsketch}
The proof is more subtle and constructs an MSO formula expressing the parity of
the number of facts of a unary predicate, using a second auxiliary relation.
The lower bound for parity comes again
from~\cite{wegener1987complexity}.
\end{proofsketch}

We leave open the question of whether this bound can be further improved, but we
note that even an $\Omega(\card{I}^3)$ lower bound would require new
developments in the study of Boolean formulae. Indeed, the best currently
known lower bound on formula   size, for \emph{any} explicit 
function of~$n$ variables with linear
circuits, is in $\Omega(n^{3-o(1)})$~\cite{hastad1998shrinkage}.

\begin{proof}
  Consider the signature $\sigma$ with a unary predicate $L$ and binary
  predicate $E$.
  We define the family $\I=(I_n)$ with $I_n$ having domain $\{a_1, \ldots, a_n\}$ and
  facts $L(a_i)$ for each $1 \leq i \leq n$ and $E(a_i, a_{i+1})$ for each
  $1\leq i < n$. Clearly, all instances in~$\I$ have treewidth~$1$.
  We consider the MSO formula $q$ that intuitively uses the $E$-facts to test whether the
  number of $L$-facts is odd. Formally, we define $q$ as follows, inspired by
  the definition of an automaton:
  \begin{align*}
    q\defeq \forall X_{\false} X_{\true} ~ & \text{Part}(X_{\false}, X_{\true}) \wedge
    \text{Tr}(X_{\false}, X_{\true}) \wedge
    \text{Init}(X_{\false}, X_{\true})\\
    & \quad\Rightarrow
    \forall x ~ (\neg \exists y ~ E(y, x) \Rightarrow x \in X_{\true})
  \end{align*}
  where $\text{Part}(X_{\false}, X_{\true})$ asserts that $X_{\false}$ and
  $X_{\true}$ partition the domain (where $\oplus$ denotes exclusive OR):
  \[
    \text{Part}(X_{\false}, X_{\true}) \defeq \forall x ~ (x \in X_{\false}) \oplus
    (x \in X_{\true})
  \]
  $\text{Tr}(X_{\false}, X_{\true})$ is the conjunction of the following transition
  rules, for each $b \neq b'$ in $\{\false, \true\}$:
  \begin{align*}
    \forall x y ~ E(x, y) \wedge y \in X_{b} \wedge L(x) & \Rightarrow x \in
    X_{b'}\\
    \forall x y ~ E(x, y) \wedge y \in X_{b} \wedge \neg L(x) & \Rightarrow x
    \in X_{b}
  \end{align*}
  and $\text{Init}(X_{\false}, X_{\true})$ asserts the initial states:
  \begin{align*}
    \forall x ~ (\neg \exists y ~ E(x, y)) \wedge \neg L(x) & \Rightarrow x \in
    X_{\false}\\
    \forall x ~ (\neg \exists y ~ E(x, y)) \wedge L(x) & \Rightarrow x \in
    X_{\true}
  \end{align*}
  Intuitively, on any possible world of~$I_n$ where all $E$-facts are present,
  $q$ is true whenever the number of $L$-facts is odd. Indeed, it is clear that
  there is a unique choice of $X_{\false}$ and $X_{\true}$ in such worlds, defined by putting $a_n$
  in $X_{\true}$
  or $X_{\false}$ depending on whether $L(a_n)$ holds, and, for $1 \leq i < n$,
  letting $b$ such that $a_{i+1} \in X_b$ and $b'$ be~$\true$ or~$\false$
  depending on whether $L(a_i)$ holds or not, putting $a_i$ in $X_{b \oplus b'}$. Hence, in worlds containing all
  $E$-facts, there is a unique choice of
  $X_{\false}$
  and $X_{\true}$ where we have $a_i \in X_{\true}$ iff the number of facts $L(a_j)$ with $i
  \leq j \leq n$ is odd. Hence, $q$ is satisfied iff, in this unique assignment,
  $a_1$ (the only node with no incoming $E$-edge) is in $X_{\true}$, that is, if the
  overall number of $L$-facts is odd.

  Hence, pick $I \in \I$ and let $\psi$ be a formula representation of
  the lineage of~$q$ on~$I$.
  Replacing the input gates for the $E$-facts by constant $1$-gates, we obtain a
  formula of the same size that computes the parity function of the inputs
  corresponding to the $L$-gates, the number of which is $\lceil \card{I}/2
  \rceil$.
  
  Now, by Theorem 8.2, Chapter~8 of
  \cite{wegener1987complexity}, any formula using
  $\land$, $\lor$, $\lnot$ expressing
  the parity function over $n$ variables has a number of variable
  occurrences that is at least~$n^2$. Hence, we deduce that $\card{\psi} =
  \Omega(\card{I}^2)$, as claimed.
\end{proof}

\section{OBDD Size Bounds}\label{sec:meta}
We have shown in the previous sections that MSO queries on bounded-treewidth
instances have tractable lineage representations as circuits and OBDDs.
This section focuses on OBDDs and shows our \emph{second main dichotomy result}:
bounded-treewidth is necessary for MSO query lineages to have polynomial OBDDs.

We first state this result in Section~\ref{sec:ucqneq-dichotomy}. Its upper
bound is Theorem~\ref{thm:getobdd}, and its lower bound
applies to a specific \ucqneq $q_{\p}$ (which only depends on the signature).
We show that $q_{\p}$ has no polynomial-width
OBDDs on \emph{any} arity-2 instance family with treewidth densely unbounded
poly-logarithmically. This second dichotomy result thus shows that 
bounded-treewidth is necessary for some \ucqneq queries to have tractable OBDDs;
it applies to a more restricted class than the FO query of
our first main dichotomy result (Theorem~\ref{thm:dichotomy}),
but applies to a different task (the computation of OBDD lineages, rather than
probability evaluation).

We then study in Section~\ref{sec:connected-ucqneq-meta-dichotomy} the 
language of \emph{connected} \ucqneq. For this language, we show that queries
can be classified in a \emph{meta-dichotomy} result: we characterize the
\emph{intricate} queries, such as $q_{\p}$, which have no polynomial OBDDs on any unbounded
treewidth family in the sense above; and we show that non-intricate
queries actually have \emph{constant-width} OBDDs on some well-chosen
unbounded-treewidth instance family. Hence, if a connected \ucqneq has
polynomial OBDDs on some unbounded-treewidth instance family, then it must have
constant-width OBDDs on some other such family.

Finally, we investigate in Section~\ref{sec:more-restricted-classes} whether our second
dichotomy result holds for more restricted fragments than \ucqneq.
First, we show that connected \cqneq queries are never intricate, so we cannot
show our dichotomy result with such queries. Second, we show the same for
connected \ucq; in fact, we show that no query \emph{closed under homomorphisms}
could be used.
We last show that our meta-dichotomy fails for disconnected queries.

As in Sections~\ref{sec:probability} and~\ref{sec:consequences},
we limit ourselves to arity-2 signatures in this section.

\subsection{A Dichotomy on OBDD Size}\label{sec:ucqneq-dichotomy}
This section shows that our Theorem~\ref{thm:getobdd} on the existence of
tractable lineage representations as OBDDs is unlikely to extend to milder conditions than
bounded-treewidth. Indeed, there are even \ucqneq queries that have no
polynomial-width OBDDs on any unbounded-treewidth input instance
with treewidth densely unbounded
poly-logarithmically, again on arity-two signatures.
Here is our \emph{second main dichotomy result} which shows this:

\begin{theorem}\label{thm:dichoobdd}
  There exists a constant $d \in \NN$ such that the following holds.
  Let $\sigma$ be an arbitrary arity-2 signature and $\I$ be a class of
  $\sigma$-instances. Assume there is a function $f(k) =
  O\big(2^{k^{1/d}}\big)$ such that,
  for all $k \in \NN$, if\/ $\I$ contains instances of treewidth $\geq
  k$, one of them has size $\leq f(k)$.
  We have the following dichotomy:
  \begin{itemize}
  \setlength\itemsep{0pt}
    \item If there is $k\in\mathbb{N}$ such that $\tw(I)\leq k$
      for every $I\in\I$, then for every MSO
      query~$q$, an OBDD of
      $q$ on~$I$ can be computed 
      in time polynomial in~$\card{I}$.
    \item Otherwise, there is a \ucqneq query $q_{\p}$ (depending on $\sigma$ but
      not on $\I$) such that the width of
      any OBDD of $q_{\p}$ on~$I\in\I$ cannot be bounded by any
      polynomial in~$\card{I}$.
  \end{itemize}
\end{theorem}

This does \emph{not} require treewidth-constructibility, and imposes
instead a slight weakening\footnote{The condition is weaker because we require the
subexponentiality to work for some fixed~$d$, not an arbitrary~$c$.}
of densely unbounded poly-logarithmic treewidth.
It does not require $\I$ to be subinstance-closed either, unlike in
Section~\ref{sec:consequences}.

The first part of the theorem is by Theorem~\ref{thm:getobdd}, so we 
sketch the proof of the second part.
Our choice of \ucqneq~$q_{\p}$
intuitively tests the existence of
a path of length~$2$ in the Gaifman graph of the instance, i.e., a violation of
the fact that the possible world is a matching of the original instance. 
Again, while we know that probability evaluation for $q_\p$ is $\fpsp$-hard if we allow
\emph{arbitrary} input instances (as counting matchings reduces to it), our task
is to show that $q_\p$ has no polynomial-width OBDDs when restricting to
\emph{any}
instance family that satisfies the conditions, a much harder task.

To show this, we draw a link between treewidth and OBDD width for $q_{\p}$
on \emph{individual} instances, with the following result (which is specific
to~$q_{\p}$):

\begin{lemma}\label{lem:obddlower}
  Let $\sigma$ be an arity-2 signature.
  There exist constants~$d', k_0 \in \NN$
  such that for any instance $I$ on~$\sigma$
  of treewidth $\geq k_0$,
  the width of an OBDD for $q_{\p}$ on~$I$ is
  $\geq 2^{(\tw(I))^{1/d'}}$.
\end{lemma}

\begin{proofsketch}
  The proof is technical and uses Lemma~\ref{lem:extract} to extract
  high-treewidth topological minors of a specific shape: \emph{skewed grids}.
  We then show that
  any variable order that enumerates the
  edges of the skewed grid must have a prefix that
  \emph{shatters} the grid, i.e.,
  sufficiently many independent grid nodes have
  both an enumerated incident edge and a non-enumerated one.
  This forces any OBDD for~$q_\p$ to remember the exact
  configurations of the enumerated edges at the level for this shattering prefix
  of its variable order. Thus, the OBDD has superpolynomial width.
  We formalize this via the structure of the prime implicants of the lineage.
\end{proofsketch}

\subsection{A Meta-Dichotomy for UCQ$^{\neq}$}\label{sec:connected-ucqneq-meta-dichotomy}
For which queries does Theorem~\ref{thm:dichoobdd} adapt? 
It does not extend to all
\emph{unsafe}~\cite{dalvi2012dichotomy} queries,
as a query
may be unsafe and still be tractable on some unbounded-treewidth instance
family: for instance, the
standard unsafe query $R(x) \land S(x, y) \land T(y)$ from~\cite{dalvi2007efficient} has
trivial OBDDs on the family of $S$-grids without unary
relations.

We answer this question, again on arity-2 signatures, by introducing a
notion of \emph{intricate} queries. We show that it precisely characterizes the
\emph{connected} \ucqneq queries for which the dichotomy of
Theorem~\ref{thm:dichoobdd} applies. Let us first recall the definition of \emph{connected}
\ucqneq queries:

\begin{definition}
  \label{def:qconnect}
  A \cqneq is \deft{connected} if, building the graph~$G$ on its atoms that connects those that
  share a variable (ignoring $\neq$-atoms), $G$ is connected (in particular it
  has no isolated vertices, unless it consists of a single isolated vertex).
  A \ucqneq is \deft{connected} if all its \cqneq disjuncts are
  \emph{connected}.
\end{definition}

We now give our definition of \defo{intricate} queries. We characterize
them by looking at \emph{line instances}:

\begin{definition}
  A \deft{line instance} is an instance~$I$ of the following form: a domain $a_1,
  \ldots, a_n$, and, for $1 \leq i < n$, one single binary fact between $a_i$ and
  $a_{i+1}$: either $R(a_i, a_{i+1})$ for some $R \in \sigma$ or $R(a_{i+1},
  a_i)$ for some binary $R \in \sigma$. (Recall that, as $\sigma$ is
  \emph{arity-two}, its maximal arity is two, so it must include at least one
  binary relation.)
\end{definition}

The intuition is that a query is intricate if, on any
sufficiently long line instance,
it must have a minimal match
that contains the two middle facts (i.e., the ones that are incident to the
middle element). Here is the formal definition of \emph{intricate} queries:

\begin{definition}
  A \ucqneq $q$ is \deft{$n$-intricate} for $n \in \NN$ if, for every line
  instance $I$ with $\card{I} = 2n + 2$, letting $F$ and $F'$ be the two facts
  of~$I$ incident to the middle element $a_{n+2}$,
  there is a \emph{minimal} match of~$q$ on~$I$ that includes both~$F$ and~$F'$.

  We call $q$ \deft{intricate} if it is $\card{q}$-intricate.
\end{definition}

Observe that queries $q$ with $\card{q} < 2$ clearly cannot be intricate.
Further, if a query has no matches that include only binary facts, then it
cannot be intricate; in other words, any disjunct that contains an atom for a
unary relation can be ignored when determining whether a query is intricate.
By contrast, our query $q_{\p}$ of
Theorem~\ref{thm:dichoobdd} was designed to be intricate, in fact $q_{\p}$ is
$0$-intricate. Also note that
an $n$-intricate query is always $m$-intricate for any $m > n$:
consider the restriction of any line instance of size $2m+2$ to a line
instance of size $2n+2$, and find a match in the restriction.

We note that we can decide whether \ucqneq queries are intricate or not, by
enumerating line instances. We do not know the precise complexity of this task:

\begin{lemma}\label{lem:intricate-decidable}
  Given a connected \ucqneq $q$, we can decide in PSPACE whether $q$ is
  intricate.
\end{lemma}

We can now state our \emph{meta-dichotomy}: a dichotomy such as
Theorem~\ref{thm:dichoobdd} holds for a connected
\ucqneq~$q$ if and only if it is intricate. Further, \emph{non-intricate}
queries must actually have \emph{constant-width} OBDD on some counterexample
unbounded-treewidth family:

\begin{theorem}\label{thm:metadicho}
  For any connected \ucqneq $q$ on an arity-2 signature:
  \begin{itemize}
  \setlength\itemsep{0pt}
    \item If $q$ is not intricate, there is a treewidth-constructible and
      unbounded-treewidth family $\I$ of instances such that $q$ has
      \emph{constant-width} OBDDs on~$\I$; the OBDDs can be computed
      in PTIME from the input instance.
    \item If $q$ is intricate, then Theorem~\ref{thm:dichoobdd} applies to $q$:
      in particular, for any unbounded-treewidth family $\I$ of instances satisfying
      the hypotheses, $q$ does not have polynomial-width OBDDs on $\I$.
  \end{itemize}
\end{theorem}

\begin{proofsketch}
We construct $\I$ for
non-intricate \ucqneq $q$ as a family of grids from a
line instance which is a counterexample to intricacy.
As we can disconnect facts that do not co-occur in a match,
we can disconnect the grids to
bounded-pathwidth instances in a lineage-preserving fashion.

Conversely, we adapt the
hardness proof of Theorem~\ref{thm:dichoobdd} to any
intricate \ucqneq query $q$, extracting independent matches from any
sufficiently subdivided skewed grid minor thanks to intricacy.
\end{proofsketch}

\subsection{Other Query Classes}\label{sec:more-restricted-classes}
We finish by investigating the status of other query classes relative
to our meta-dichotomy, to see whether 
Theorem~\ref{thm:dichoobdd} could be shown for
queries in an even less expressive class than \ucqneq, such as \cqneq or \ucq.

\paragraph*{Connected \cqneq queries}
We classify the connected \cqneq queries relative to
Theorem~\ref{thm:metadicho}, by showing that a connected \cqneq can never
be intricate. This explains why, for instance, the query $R(x) \land S(x, y)
\land T(y)$
is not intricate, as is witnessed by the family of $S$-grids.

\begin{proposition}\label{prp:cqneqnotintricate}
  A connected \cqneq is never intricate.
\end{proposition}

\begin{proofsketch}
  The signature $\sigma$ must contain at most one binary relation $R$, as otherwise we can
  find for any
  \cqneq $q$ a
  family of grids where the query never holds, so that $q$ would have trivial constant-width
  OBDDs. Now, if $q$ contains a join pattern of the form
  $R(x, y) \land R(y, z)$, then $q$ has no matches on line instances with
  $R$-facts of alternating directions. If
  $q$ does not contain such a pattern, we consider line instances with a path
  of~$R$-facts in the same direction, and show that $q$ has no match that
  involves
  the two middle facts.
\end{proofsketch}

By Theorem~\ref{thm:metadicho}, this implies that any \cqneq query $q$ has an
unbounded-treewidth, treewidth-constructible family of instances $\I$ such that
$q$ has constant-width OBDDs on $\I$ (that can be computed in PTIME); and it
also implies that we could not have proven Theorem~\ref{thm:dichoobdd} with a
connected \cqneq query.

\paragraph*{Homomorphism-closed}
Second, we investigate the status in our meta-dichotomy of queries without
inequalities, i.e., connected \ucq{}s. We can in fact
show a result for all queries that are 
\emph{closed under homomorphisms}, no matter whether they
are connected or not. Further, we can even choose a \emph{single} class of instances which is easy
for \emph{all} query closed under homomorphisms. (Remember that our queries are
always constant-free.)

\begin{proposition}\label{prp:nodichohomom}
  For any arity-2 signature,
  there is a
    tree\-width-constructible instance family $\I$ with unbounded treewidth
    and $w\in\NN$
      such that any query $q$ closed under homomorphisms has OBDDs of
      width~$w$ on~$\I$ that can be computed in PTIME in the input instance.
\end{proposition}

\begin{proofsketch}
  Queries $q$ closed under homomorphisms become essentially
  trivial on the class $\I$ of complete bipartite directed graphs: all minimal
  matches of $q$ on $\I$ (if any) have a single fact, hence $q$ has a
  constant-width OBDD.
\end{proofsketch}

Hence, a connected \ucq is never intricate,
so we could not have shown Theorem~\ref{thm:dichoobdd} with a
\ucq query rather than a \ucqneq query.

Again, this result should not be confused with those
of~\cite{jha2012tractability,jha2013knowledge}. Of course, not all
homomorphism-closed queries, or even
\ucq{}s, have constant-width OBDDs on arbitrary instances. We are merely
claiming the \emph{existence} of high-treewidth instance classes for which
we have constant-width OBDDs whatever the query.

\paragraph*{Beyond connected queries}
We consider last whether our dichotomy in Theorem~\ref{thm:dichoobdd} could
extend to \emph{disconnected} \cqneq, which are not covered by
Proposition~\ref{prp:cqneqnotintricate}
or by the meta-dichotomy of Theorem~\ref{thm:metadicho}.

If the signature has more than one binary relation, this is
hopeless: the easy argument used in the proof 
of Proposition~\ref{prp:cqneqnotintricate} in this case can also apply to disconnected
\cqneq.

However,
quite surprisingly,
on
signatures with a \emph{single} binary relation (and arbitrarily many unary
ones) we can show a weakening of Theorem~\ref{thm:dichoobdd} for a disconnected
\cqneq. The first part adapts (it holds for all MSO), so only the lower bound
is interesting, which we can rephrase as before to a lower bound on OBDD width
on individual input instances:

\begin{proposition}\label{prp:discq}
  Let $\sigma$ be an arity-2 signature with only one binary relation.
  There exists a disconnected \cqneq{} query $q_\dd$, a constant~$d'>1$ and
  integer $n_0\in\NN$ such that:
  for any instance $I$ on~$\sigma$ of size
  $\geq n_0$,
  letting $k$ be the treewidth of~$I$, the width of any OBDD for $q_\dd$ is
  $\Omega(k^{1/d'})$.
\end{proposition}

\begin{proofsketch}
The query $q_\dd$ tests whether there are at least two facts with disjoint
domains. We adapt Lemma~\ref{lem:obddlower} and use a skewed grid minor of the
instance, but show this time that the OBDD must remember linearly many
configurations at some level.
\end{proofsketch}

This implies that $q_\dd$ does not satisfy the first part of the meta-dichotomy
of Theorem~\ref{thm:metadicho}.
Surprisingly, however, we can show that the query $q_\dd$ has OBDDs of width
$O(k)$ on some unbounded-treewidth and treewidth-constructible instance class.
Hence, $q_\dd$ does not satisfy the second part of the meta-dichotomy either, so
$q_\dd$ witnesses that there are \emph{disconnected} \cqneq which do not follow our
meta-dichotomy at all!
We leave to future work a more precise study of disconnected queries.

\section{Connection to Safe Queries}
\label{sec:safe}
We conclude this paper by connecting our results to \emph{query-based} tractability conditions. More
specifically, we focus on UCQs that have polynomial OBDD representations of their
lineage: by the results of~\cite{jha2013knowledge}, those are the
\emph{inversion-free UCQs}. We will show that the tractability of such queries
can be explained by our \emph{data-based} tractability conditions: more
precisely, for any inversion-free UCQ, there is a lineage-preserving rewriting of
input instances to instances that have constant \emph{tree-depth}~\cite{nesetril2012bounded},
and hence (by Lemma~11 of~\cite{bodlaender1995approximating}) have constant
pathwidth and treewidth. The definition of \emph{tree-depth} is as follows:

\begin{definition}
  An \emph{elimination forest} for an (undirected) graph $G$ is a forest $F$ on the
  vertices of~$G$ such that, for any edge $\{x, y\}$ of $G$, one of $x$ and $y$
  is a descendant of the other in~$F$.
  The \emph{tree-depth} of~$G$ is the minimal height of an elimination forest
  of~$G$.
  The \emph{tree-depth} of an instance $I$ is that of its Gaifman graph.
\end{definition}

This section applies to
signatures of arbitrary arity.

\paragraph*{Unfoldings}

Our results are based on instance rewritings of a general kind, possibly of
independent interest.
We let~$I$ denote an arbitrary instance in
this paragraph, and let $q$ denote a query closed under homomorphisms.

\begin{definition}
  An \emph{unfolding} of instance~$I$ is an instance~$I'$ with a homomorphism
  $h$ to $I$ which is \emph{bijective on facts}: for any fact $F(\mathbf{a})$ of $I$,
  there is exactly one fact $F(\mathbf{a}')$ in~$I$ such that $h(a'_i) = a_i$
  for all~$i$.
\end{definition}

The bijection defined by the homomorphism allows us to see the lineage
of $q$ on an unfolding $I'$ of~$I$ as a Boolean function on the same variables as the
lineage of~$q$ on~$I$.

We use unfoldings as a tool to show lineage-preserving instance
rewritings.
Indeed, we can see from the homomorphism $h$ from $I'$ to $I$ that
any match of
$q$ in $I'$ is preserved in~$I$ through~$h$. In other words, the following is
immediate:

\begin{lemma}
  \label{lem:onedir}
  If $I'$ is an unfolding of~$I$
  and $\phi$ and $\phi'$ are the lineages of $q$ on $I$ and $I'$,
  then for any valuation $\nu$ of the facts of~$I$,
  if $\nu(\phi') = \true$ then $\nu(\phi) = \true$.
\end{lemma}

The converse generally fails, but a sufficient condition is:

\begin{definition}
  An unfolding $I'$ of~$I$ \emph{respects} $q$ if, for any
  match $M \subseteq I$ of~$q$ on~$I$,
  letting $M'$ be its preimage in~$I'$, we have $M' \models q$.
\end{definition}

Intuitively, the unfolding does not ``break'' the matches of~$q$. This ensures
that the lineage is preserved exactly:

\begin{lemmarep}\label{lem:samelin}
  If $I'$ is an unfolding of~$I$ that respects~$q$,
  then $q$ has the same lineage on~$I$ and~$I'$.
\end{lemmarep}

\begin{proof}
  By Lemma~\ref{lem:onedir}, it suffices to show that for any match $M$ of $q$ in
  $I$, the preimage $M'$ of $M$ by the bijection on facts is also a match of
  $q$; but this is precisely what is guaranteed by the fact that $I$
  respects~$q$.
\end{proof}

\paragraph*{Inversion-free UCQs}

We use unfoldings to study Boolean constant-free \emph{inversion-free} UCQ
queries. We do not restate their formal definition here, and refer the reader to
Section~2 of~\cite{jha2013knowledge}.
The following is known:

\begin{theorem}[(Proposition~5 of~\cite{jha2013knowledge})]
  \label{thm:jhainversion}
  For any inversion-free UCQ~$q$, for any input
  instance $I$, the lineage of $q$ on $I$ 
  has an OBDD of constant width (i.e., the width only depends on~$q$).
\end{theorem}

When studying inversion-free UCQs, it is convenient to assume that the
\emph{ranking} transformation was applied to the query and instance~\cite{dalvi2010computing,dalvi2012dichotomy}. A UCQ is
\emph{ranked} if, defining a binary relation on its variables by setting $x <
y$ when $x$ occurs before $y$ in some atom, then $<$ has no cycle.
In particular, in a ranked query, no variable occurs twice in an atom. An
instance is \emph{ranked} if there is a total order $<$ on its domain such that for
any fact $R(\mathbf{a})$ and $1 \leq i < j \leq \arity{R}$, we have $a_i < a_j$.
In particular, no element occurs twice in a fact. Up to changing the signature,
we can always rewrite a UCQ $q$ to a ranked UCQ $q'$, and rewrite separately any
instance~$I$ to a ranked instance $I'$, so that the lineage of~$q$ on~$I$ is the
same as that of $q'$ on $I'$;
see~\cite{dalvi2010computing,dalvi2012dichotomy} for details.

We will thus assume that the ranking transformation has been applied to the query, and
to the instance. Note that this can be performed in linear time in
the instance, and does not change its treewidth, pathwidth, or tree-depth, as the Gaifman
graph is unchanged by this operation.

Once this ranking transformation has been performed, we can show the following:

\begin{theoremrep}\label{thm:unfoldsafe}
  For any ranked inversion-free UCQ~$q$, 
  for any ranked instance $I$, there is an unfolding $I'$ of~$I$ that respects $q$ and has
  tree-depth $\leq \arity{\sigma}$.
\end{theoremrep}

Hence, in particular, $q$ has the same lineage on~$I'$ and on~$I$,
as shown by Lemma~\ref{lem:samelin}.
As pathwidth is less than tree-depth~\cite{bodlaender1995approximating},
by Theorem~\ref{thm:pathwidth}, this implies
the result of Theorem~\ref{thm:jhainversion}, and (via
Proposition~\ref{prp:makecircuitpath}) generalizes it
slightly: it shows that 
the lineage
can even be represented by a bounded-pathwidth circuit.

\begin{proofsketch}
  We use an inversion-free expression~\cite{jha2013knowledge} for $q$ to define
  an order on relation attributes which is compatible across relations. We
  unfold each relation by distinguishing each element depending on the tuple of
  elements on the preceding positions; this is inspired by Proposition~5
  of~\cite{jha2013knowledge}. The result preserves the inversion-free expression
  and has a natural elimination tree defined by the prefix ordering on the
  distinguished elements.
\end{proofsketch}

Theorem~\ref{thm:unfoldsafe} thus suggests that the tractability of probability evaluation
for inversion-free UCQs can be understood in terms of bounded-tree-depth and bounded-pathwidth
tractability:
what inversion-free UCQs ``see'' in an instance is a bounded tree-depth
structure.

\begin{toappendix}
The roadmap of the proof is as follows.
We use an \emph{inversion-free expression}~\cite{jha2013knowledge}
for $q$ to define an order on relation attributes which is compatible across
relations. We then unfold each relation by distinguishing each element depending
on the tuple of elements on the preceding positions; this is inspired by
Proposition~5 of~\cite{jha2013knowledge}. The result preserves the
inversion-free expression and has a path decomposition that enumerates the facts
lexicographically.
To follow the roadmap, we first define inversion-free expressions as in~\cite{jha2013knowledge}:

\begin{definition}
  A \emph{hierarchical expression}~\cite{jha2013knowledge} is a logical sentence
  built out of atoms, conjunction, disjunction, and existential quantification, where
  each variable is a \emph{root variable}, i.e., occurs in all atoms in the
  scope of its existential quantifier.

  An \emph{inversion-free expression} is a hierarchical expression such that, for each
  relation symbol $R$, we can define a total order $<_R$ on its positions
  $\{R^1, \ldots, R^{\arity{R}}\}$, such that, in every $R$-atom $R(\mathbf{x})$,
    if $R^i <_R R^j$ then the quantifier $\exists x_j$ in the query is in the
    scope of the quantifier $\exists x_i$.
\end{definition}

By Proposition~2 of~\cite{jha2013knowledge}, a ranked UCQ is inversion-free iff
it can be written as an inversion-free expression, so it suffices to show
Theorem~\ref{thm:unfoldsafe} for inversion-free expressions.

\medskip

We first define our unfolding $I'$ of an input instance $I$. For each fact
$R(\mathbf{a})$ of~$I$, we create the fact $R(\mathbf{b})$ defined as follows.
Writing $R^{i_1} <_R \cdots <_R R^{i_n}$ the positions of~$R$ according to the
total order $<_R$, we define $b_{i_1}$ as the tuple $(a_{i_1})$, and define
$b_{i_j}$ as the tuple formed by concatenating $b_{i_{j-1}}$ and $(a_{i_j})$.
We call $f_R$ the operation thus defined, with $\mathbf{b} = f_R(\mathbf{a})$.
Clearly the operation $h$ mapping each tuple to its last element is a homomorphism
from~$I'$ to~$I$, and it is bijective on facts because it is the inverse of the
operation that we described. Hence, $I'$ is an unfolding of~$I$.
Note that this construction is similar to the one used in the proof of Proposition~5
in~\cite{jha2013knowledge}.

\medskip

We must show that $I'$ has bounded tree-depth. To do this, consider the
elimination forest $F$ defined on $\dom(I')$ by setting $\mathbf{b}$ to be the
parent of~$\mathbf{c}$ iff $\mathbf{b}$ is a longest strict prefix
of~$\mathbf{c}$. The forest $F$ has one root per singleton element
in~$\dom(I')$: these elements correspond the (possibly strict) subset
of~$\dom(I)$ of the elements occurring at the first position $R^{i_1}$ for the
order $<_R$ for some relation $R$. It is clear that $F$ is indeed an
elimination forest for the Gaifman graph of~$I'$, as by construction any fact
$R(\mathbf{b})$ of~$I'$ is such that, letting $n \defeq \arity{R}$ and $R^{i_n}$
be the last position of~$R$ in the order~$<_R$, the elements of~$\mathbf{b}$ are
exactly the non-empty prefixes of~$b_{i_n}$, so, for any pair $b_i$, $b_j$ of
elements of~$\mathbf{b}$,
one is a prefix of the other, so one is an ancestor of the other in~$F$. We
conclude by noticing that the elimination forest~$F$ has height 
$\arity{\sigma}$, so the tree-depth of~$I'$ is at most $\arity{\sigma}$.

\medskip

The only thing left to show is that $I'$ respects~$q$. For this, let us consider
the inversion-free expression $Q$ of~$q$. For any subexpression $\phi$ of
$Q$ with free variables~$\mathbf{x}$, let us define the \emph{ordered free variables}
of $\phi$, denoted $\ofv(\phi)$, as follows. If $\phi$ contains no atoms (i.e.,
it is the constant formula ``$\text{true}$'' or ``$\text{false}$''), then
$\mathbf{x}$ is empty and so is $\ofv(\phi)$. Otherwise, as $Q$ is
inversion-free, it is in particular hierarchical, so all free variables of
$\phi$ must occur in all atoms of $\phi$: this is by definition, for any free
variable $x_i$ of~$\phi$,
of the subexpression of~$Q$ that includes $\phi$ whose outermost operator is
$\exists x_i$. Hence, consider any atom $A = R(\mathbf{x})$, and, remembering
that no variable occurs twice in~$A$ (as $Q$ is ranked), define
$\ofv(\phi)$ as the total order on $\mathbf{x}$ given by $x_i < x_j$ iff $R^i
<_R R^j$.

It is clear that $\ofv(\phi)$ is well-defined, i.e., that does not depend on our choice of
atom in $\phi$: this is because $Q$ is an inversion-free expression, so the
order of variables in atoms must reflect the order in which the variables are
quantified.

We now show the claim that $I'$ respects $Q$:

\begin{lemma}
  If $I$ has a match $M$ of $Q$, then, defining $M'$ by mapping each fact
  $R(\mathbf{a})$ of $M$ to the fact $R(f_R(\mathbf{a}))$ of~$I'$, $M'$ is a
  match of~$Q$ in~$I'$.
\end{lemma}

\begin{proof}
  Let $f$ be defined on tuples of $\dom(I)$ by $f(\mathbf{a}) \defeq (a_1, (a_1,
  a_2), \ldots, \mathbf{a})$.
  For any subformula $\phi$ with $n$ free variables and any $n$-tuple $\mathbf{a}$
  of $\dom(I)$, we write $I \models \phi[\ofv(\phi) {\defeq} \mathbf{a}]$ to mean
  the Boolean formula with constants obtained by substituting each variable in
  $\ofv(\phi)$ by the corresponding element in $\mathbf{a}$ following the order
  of $\mathbf{a}$ and $\ofv(\phi)$.

  We proceed by induction on the subformulae of~$Q$, showing that if a
  subformula~$\phi$ and tuple $\mathbf{a} \in \dom(M)$ is such that $M \models
  \phi[\ofv(\phi) {\defeq} \mathbf{a}]$, then $M' \models
  \phi[\ofv(\phi){\defeq}
  f(\mathbf{a})]$.
  \begin{itemize}[leftmargin=9pt]
    \item For atoms, this is by definition of $\ofv$ and of~$M'$.
    \item For $\phi \wedge \psi$, we observe that we have $\ofv(\phi) =
      \ofv(\phi \wedge \psi) = \ofv(\psi)$: write $\mathbf{x}$ to refer to these
      ordered free variables.
      If $M \models (\phi \wedge \psi)[\mathbf{x} {\defeq} \mathbf{a}]$,
      then $M \models \phi[\mathbf{x} {\defeq} \mathbf{a}]$ and $M \models
      \psi[\mathbf{x} {\defeq} \mathbf{a}]$, as by induction $M' \models
      \phi[\mathbf{x} {\defeq} f(\mathbf{a})]$ and $M' \models \psi[\mathbf{x}
      {\defeq} f(\mathbf{a})]$, we deduce $M' \models (\phi \wedge \psi)[\mathbf{x}
      {\defeq} f(\mathbf{a})]$.
    For $\phi \vee \psi$, the reasoning is the same.
    \item For $\phi: \exists y ~ \psi$, writing $\mathbf{x} {\defeq} \ofv(\phi)$,
      by definition of~$\ofv$, $y$ is the last variable
      of $\mathbf{x}' {\defeq} \ofv(\psi)$. As $M \models \phi[\mathbf{x}
      {\defeq} \mathbf{a}]$, by
      definition there is $c \in \dom(M)$ such that, letting $\mathbf{a}'$ be
      the concatenation of $\mathbf{a}$ and $c$,
      $M \models \psi[\mathbf{x}' {\defeq} \mathbf{a}']$. By induction
      hypothesis we have $M' \models \psi[\mathbf{x}' {\defeq} f(\mathbf{a}')]$, and as removing the last element of $f(\mathbf{a}')$ yields
      $f(\mathbf{a})$, we deduce that $M' \models (\exists y ~
      \psi)[\mathbf{x} {\defeq}
    f(\mathbf{a})]$.
  \end{itemize}
  The outcome of this induction is that $M \models Q$ implies $M' \models
  Q$, the desired claim.
\end{proof}
\end{toappendix}

\section{Conclusion}
\label{sec:conclusion}
The main result of this work justifies that bounded treewidth
is the right condition on instances to make
probability evaluation tractable, with a dichotomy between \emph{ra-linear} evaluation for
\emph{MSO}
queries assuming bounded treewidth, and \emph{$\fpsp$-hardness} under RP
reductions for \emph{FO queries}
otherwise. Our second main result extends this to \ucqneq for tractable OBDD
representations, specifically, to the \emph{intricate} queries that we
characterize: they have no polynomial OBDDs on \emph{any} sufficiently dense unbounded-treewidth
instance class.

We do not know whether our results extend beyond
arity-2, e.g., using techniques from the CSP
context~\cite{marx2013tractable}. Another question is whether the first
dichotomy result extends to
more restricted query classes, or even to the \ucqneq $q_\p$ of Theorem~\ref{thm:dichoobdd}:
indeed, probability evaluation of~$q_\p$ is \#P-hard
but we were unable to show this \emph{under arbitrary
subdivisions}.
Another extension would be to use PTIME rather than RP reductions,
which would follow from a derandomization
of~\cite{chekuri2014polynomial}.
In terms of lineage, we do not know either whether our OBDD results extend
to other lineage classes, e.g., FBDDs or d-DNNFs.

Our main hope, though, concerns the
\emph{unfolding} technique of Section~\ref{sec:safe}. We showed that, for
inversion-free UCQs, unfolding 
\emph{always} reduces an instance to a bounded-treewidth one while preserving 
lineage.
Could there be a \emph{query-dependent}, lineage-preserving unfolding of
instances, lowering the treewidth
by undoing joins that the query does not
``see''? Such a technique could yield a tractability
condition on the instance \emph{and} query, covering and extending \emph{both}
bounded-treewidth and safe queries.
It could also be practically useful to approximate query
probabilities, maybe in conjunction with the query-based 
\emph{dissociation} technique~\cite{gatterbauer2015approximate}.

\paragraph*{Acknowledgements}
We thank Mikaël Monet for careful proofreading, and
Chandra Chekuri and Dan Suciu for technical clarifications
on~\cite{chekuri2014polynomial} and~\cite{jha2013knowledge}, respectively.
We also thank the anonymous reviewers of PODS~2016 for their valuable comments,
in particular for strengthening Theorem~\ref{thm:unfoldsafe} to tree-depth.
This work was partly funded by the Télécom ParisTech Research Chair on
Big Data and Market Insights.

\bibliographystyle{abbrv}
\bibliography{references}

\end{document}